\documentclass[11pt]{article}

\usepackage[margin=1in]{geometry}
\usepackage[T1]{fontenc}
\usepackage{lmodern}
\usepackage{amsmath,amssymb,amsthm,mathtools}
\usepackage{enumitem}
\usepackage{microtype}
\usepackage[dvipsnames]{xcolor}
\usepackage[colorlinks=true,linkcolor=BrickRed,citecolor=OliveGreen,urlcolor=NavyBlue]{hyperref}
\usepackage[capitalise,nameinlink]{cleveref}
\usepackage{thmtools}
\usepackage{booktabs}
\usepackage{tikz}

\declaretheorem[name=Theorem,numberwithin=section]{theorem}

\declaretheorem[name=Lemma,sibling=theorem]{lemma}
\declaretheorem[name=Corollary,sibling=theorem]{corollary}
\declaretheorem[name=Proposition,sibling=theorem]{proposition}

\declaretheorem[name=Definition,style=definition,sibling=theorem]{definition}
\declaretheorem[name=Remark,style=remark,sibling=theorem]{remark}

\newcommand{\Z}{\mathbb{Z}}
\newcommand{\Ot}{\widetilde{O}}
\newcommand{\eps}{\varepsilon}
\newcommand{\floor}[1]{\left\lfloor #1 \right\rfloor}
\newcommand{\ceil}[1]{\left\lceil #1 \right\rceil}
\newcommand{\mc}{\star}
\newcommand{\mono}{\operatorname{mono}}
\newcommand{\Zinf}{\Z\cup\{+\infty\}}
\newcommand{\set}[1]{\{#1\}}

\title{Barely Monotone $(\min,+)$-Convolution in Truly Subquadratic Time}

\author{
MohammadTaghi Hajiaghayi\thanks{University of Maryland. Email: \texttt{hajiaghayi@gmail.com}.}
\and
Danny Mittal\thanks{University of Maryland. Email: \texttt{dannymittal@gmail.com}.}
\and
Saeed Seddighin\thanks{Email: \texttt{saeedreza.seddighin@gmail.com}.}
}

\date{}

\hypersetup{pdftitle={Barely Monotone (min,+)-Convolution in Truly Subquadratic Time},
  pdfauthor={MohammadTaghi Hajiaghayi, Danny Mittal, Saeed Seddighin}}
\hypersetup{pageanchor=false}
\begin{document}

\maketitle

\begin{abstract}
The $(\min,+)$-convolution of two sequences $A$ and $B$ of length $n$ is the sequence $C$ with
$C[k]=\min_{i+j=k}(A[i]+B[j])$. For \emph{bounded} inputs, whose entries are integers in $\{0,\dots,O(n)\}$, prior
work computes it in truly subquadratic time when the inputs are monotone; the algorithm of Chi, Duan, Xie, and Zhang
(STOC 2022) takes expected $\Ot(n^{1.5})$ time. We introduce a monotonicity measure ranging from $0$ (monotone) to
$1/2$ (entirely non-monotone): a sequence has \emph{monotonicity} $\alpha$ if it can be partitioned into $O(n^\alpha)$
monotone subsequences, and by the Erd\H{o}s--Szekeres theorem every sequence has monotonicity at most $1/2$. We show
that truly subquadratic time is achievable even when just one input is barely monotone, that is, has monotonicity
$1/2-\Omega(1)$: if $A$ has monotonicity $\alpha$, we compute the convolution in expected time
$\Ot(n^{5/3+2\alpha/3})$ for every bounded $B$. If $B$ has monotonicity $\beta$ as well, the expected time improves
to $\Ot(n^{(3+\alpha+\beta)/2})$, which matches the monotone case for $\alpha=\beta=0$; this algorithm also allows
infinite entries placed arbitrarily.

We complement these algorithms with fine-grained reductions. Bounded $(\min,+)$-convolution reduces to bounded
monotone $(\min,+)$-convolution of length $N=O(n^{1.5})$, so an $O(N^{4/3-\eps})$-time algorithm for monotone
inputs would give an $O(n^{2-3\eps/2})$-time algorithm for bounded inputs. Similarly, entries bounded by $n$ reduce
to entries bounded by $N^x$ on sequences of length $N=\Theta(n^{2/(1+x)})$. We also show that if only $A$ has entries in
$\{0,\dots,M\}$, we can compute the convolution in $\Ot(n(M+1))$ time, and in $\Ot(n^{1.5}\sqrt M)$ time if $A$ may
also contain $+\infty$.
\end{abstract}

\thispagestyle{empty}
\newpage
\setcounter{page}{1}
\hypersetup{pageanchor=true}

\section{Introduction}
\label{sec:intro}

The All-Pairs Shortest Paths (APSP) problem is equivalent, up to subcubic reductions, to computing the
$(\min,+)$-product $C[i][j]=\min_k(A[i][k]+B[k][j])$ of two $n\times n$ matrices~\cite{VassilevskaWilliamsW18}. Until
recently the fastest known algorithm for both problems ran in $n^3/2^{\Omega(\sqrt{\log n})}$
time~\cite{Williams18}, and whether $O(n^{3-\eps})$ time is possible for some constant $\eps>0$ had been a central
open question for decades; very recently,
Alman and Vassilevska Williams~\cite{AlmanV26} gave the first truly subcubic algorithm for APSP with polynomially
bounded integer weights. The one-dimensional analogue of the $(\min,+)$-product is the
\emph{$(\min,+)$-convolution} of two sequences $A=(A[0],\dots,A[n-1])$ and $B=(B[0],\dots,B[n-1])$, namely the
sequence $C=A\mc B$ of length $2n-1$ with
\[
  C[k] \;=\; \min_{\substack{0\le i,j<n\\ i+j=k}} \bigl(A[i]+B[j]\bigr), \qquad 0\le k\le 2n-2.
\]
The two problems are closely related: as observed by Bremner et al.~\cite{BremnerCDEHILPT14}, the
$(\min,+)$-convolution of two sequences of length $n$ can be computed with $O(\sqrt n)$ $(\min,+)$-products of
$\sqrt n\times\sqrt n$ matrices. A truly subcubic algorithm for APSP therefore gives a truly subquadratic algorithm
for $(\min,+)$-convolution, which has long made the latter a natural intermediate target on the way to faster APSP
algorithms. $(\min,+)$-Convolution is also tied by fine-grained reductions to knapsack-type problems, tree
sparsity, and other problems~\cite{CyganMWW19,KunnemannPS17,BackursIS17}. The algorithm of~\cite{AlmanV26} yields a
truly subquadratic algorithm for $(\min,+)$-convolution with polynomially bounded integer entries, but with an
exponent very close to $2$. Structured inputs admit much faster algorithms, and they are the subject of this paper.

If all entries are integers in $\set{0,\dots,M}$, the convolution can be computed in $\Ot(n(M+1))$ time with the
fast Fourier transform, in the spirit of the algorithm of Alon, Galil, and Margalit~\cite{AlonGM97} for the matrix
case. For $M=\Theta(n)$, that is, for \emph{bounded} inputs with entries in $\set{0,\dots,O(n)}$, this gives no
improvement over the naive algorithm. If, in addition, both sequences are monotone, the problem becomes
substantially easier. Chan and Lewenstein~\cite{ChanL15} gave the first truly subquadratic algorithm for bounded
monotone $(\min,+)$-convolution, with expected running time $O(n^{1.859})$, using additive combinatorics, and showed
that the problem is closely related to bounded-difference $(\min,+)$-convolution and to binary jumbled indexing.
Chi, Duan, Xie, and Zhang~\cite{ChiDXZ22} then improved the expected running time to $\Ot(n^{1.5})$; we refer to
their algorithm as the CDXZ algorithm. Bringmann, D\"urr, and Polak~\cite{BringmannDP24} extended it to entries in
$\set{0,\dots,M}$, obtaining expected time $\Ot(n\sqrt{M+1})$, and showed that it suffices for one of the two
sequences to be monotone; Jin, Park, Saha, and Xu~\cite{JinPSX26} gave a deterministic version with an $n^{o(1)}$
overhead. These algorithms have found applications to knapsack problems~\cite{BringmannC22,BringmannDP24}. The
picture is similar for the $(\min,+)$-product, where bounded-difference instances and instances with bounded row- or
column-monotone matrices admit truly subcubic
algorithms~\cite{BringmannGSW16,WilliamsX20,GuPWX21,ChiDX22,ChiDXZ22,JinPSX26}.

Structured $(\min,+)$-convolution is also at the heart of our recent work on dynamic problems~\cite{HajiaghayiMS27}.
There we gave truly sublinear dynamic algorithms for $(\min,+)$-convolution and for
$(\min,+)$-matrix multiplication under \emph{marginal} updates, in which each update changes one input entry by
$\pm1$ and each query asks for one output entry. For $(\min,+)$-convolution we obtained $\Ot(n^{21/22})$ expected
amortized update time and $O(n^{21/22})$ worst-case query time, by combining the structure of bounded inputs with
additive combinatorics, even though the values themselves may grow without bound. We also reduced static
bounded-difference $(\min,+)$-convolution to this dynamic problem; update and query times of $O(n^{1/3-\eps})$ (with
$n^{2+o(1)}$ initialization time for arbitrary inputs) would improve the exponent $3/2$ for bounded-difference
inputs. For $(\min,+)$-matrix multiplication we obtained truly sublinear algorithms whenever the updates are
not too concentrated on a single row or column, and a matching hardness result based on the OMv conjecture shows
that this condition cannot be dropped. These results complement classical work on dynamic matrix inverse and dynamic
convolution.

All known algorithms for bounded inputs with running time close to $n^{1.5}$, however, rely on monotonicity, or on
closely related structure such as bounded differences, of at least one of the sequences. In this
paper we ask how much monotonicity is really needed. We measure the distance from monotonicity by the
\emph{monotonicity number} $\mono(A)$, the minimum number of monotone subsequences into which a sequence $A$ can be
partitioned. Every sequence of length $n$ has monotonicity number $O(\sqrt n)$, so we may write $\mono(A)=O(n^\alpha)$ with
$\alpha\in[0,1/2]$ and call $\alpha$ the \emph{monotonicity} of $A$; it ranges from $0$ for monotone sequences to
$1/2$ for entirely non-monotone ones, and the algorithms above cover only monotone sequences, at the extreme $\alpha=0$. We show that this
is far from necessary: bounded $(\min,+)$-convolution can be computed in truly subquadratic time even when just one
input is barely monotone, that is, has monotonicity $1/2-\Omega(1)$. We also study the
converse direction, that is, how improvements for monotone inputs can be transferred to less structured ones, and
give fine-grained reductions between bounded, monotone, and small-range instances.

\subsection{Our Results}
\label{sec:results}

\begin{table}[t]
\centering
\renewcommand{\arraystretch}{1.25}
\small
\begin{tabular}{@{}lll@{}}
\toprule
Problem & Result & Reference\\
\midrule
$\mono(A)\le n^\alpha$, $\mono(B)\le n^\beta$ & $\Ot(n^{(3+\alpha+\beta)/2})$, Las Vegas & \ref{thm:nearly-mono-1}\\
$\mono(A)\le n^\alpha$, $B$ bounded & $\Ot(n^{5/3+2\alpha/3})$, Las Vegas & \ref{thm:nearly-mono-2}\\
monotone with infinities & $\Ot(n^{1.5})$, Las Vegas & \ref{thm:nearly-mono-1}\\
$\mono(A)\le 2$, $B$ bounded & $O(n^{1.5-\eps})\Rightarrow$ bounded $O(n^{2-\eps})$ & \ref{prop:one-sided-lb}\\
$A\in[0,M]^n$, $B$ arb. & $\Ot(n(M+1))$, determ. & \ref{thm:one-bounded}\\
$A\in([0,M]\cup\set{\infty})^n$, $B$ arb., $1\le M\le n$ & $\Ot(n^{1.5}\sqrt M)$, determ. & \ref{thm:bounded-inf}\\
bounded $\to$ bounded monotone & length and values $O(n^{1.5})$ & \ref{thm:mono-red}\\
& monotone $O(N^{4/3-\eps})\Rightarrow$ bounded $O(n^{2-3\eps/2})$ & \ref{cor:four-thirds}\\
values $\le n$ $\to$ values $\le N^{x}$ & length $N=\Theta(n^{2/(1+x)})$ & \ref{thm:value-red}, \ref{cor:value-lb}\\
\bottomrule
\end{tabular}
\caption{Summary of our results. Here $n$ is the input length, $N$ is the length produced by a reduction, and
$\mono(\cdot)$ is the monotonicity number; in the value-range reduction, $0<x\le1$ is fixed. Unless stated
otherwise, finite entries are integers in $\set{0,\dots,O(n)}$, and ``arb.''\ means that $B$ may have arbitrary
integer or $+\infty$ entries. In the third row, the finite entries of each input form a monotone sequence and
$+\infty$ entries are placed arbitrarily. Las Vegas bounds are expected running times.}
\label{tab:results}
\end{table}

Throughout, $n$ denotes the length of the input sequences, $[n]=\set{0,1,\dots,n-1}$, and $\Ot(\cdot)$ suppresses
factors polylogarithmic in $n$. A sequence is \emph{monotone} if it is non-decreasing or non-increasing, and an
instance is \emph{bounded} if all entries are integers in $\set{0,\dots,O(n)}$. Our results are summarized in
\cref{tab:results}.

\emph{Barely monotone sequences.}
In the monotonicity number $\mono(A)$, the monotone subsequences need not be contiguous, and their directions may
differ. Monotone sequences have monotonicity number $1$, and by the Erd\H{o}s--Szekeres theorem~\cite{ErdosS35} every
sequence of length $n$ has monotonicity number $O(\sqrt n)$ (\cref{lem:greedy-partition}). Interleaving matters already for two subsequences: $(\min,+)$-convolution of two bounded sequences with monotonicity
number at most $2$ contains minimum-witness convolution (\cref{rem:min-witness}), for which no
$O(n^{1.5-\eps})$-time algorithm is known~\cite{LincolnPW20}, and the reduction of Bringmann, D\"urr, and Polak to
the case where both sequences are monotone does not apply. Our main result is the following.

\begin{theorem}[see \cref{thm:nearly-mono-1,thm:nearly-mono-2}]
\label{thm:intro-nearly-mono}
Let $A,B\in\set{0,\dots,O(n)}^n$ with $\mono(A)=O(n^{\alpha})$ and $\mono(B)=O(n^{\beta})$.
Then $A\mc B$ can be computed by a Las Vegas algorithm in expected time
\begin{enumerate}[label=(\alph*)]
  \item $\Ot\bigl(n^{(3+\alpha+\beta)/2}\bigr)$, and
  \item $\Ot\bigl(n^{5/3+2\alpha/3}\bigr)$, without any assumption on the monotonicity number of $B$.
\end{enumerate}
No partition into monotone subsequences needs to be given as part of the input.
\end{theorem}

Since we may assume $\alpha,\beta\le1/2$, both bounds are at most quadratic up to
polylogarithmic factors, and they are truly subquadratic under mild conditions: bound~(a) whenever $\alpha+\beta<1$, which holds in
particular when both exponents are below $1/2$, and bound~(b) whenever $\alpha<1/2$, for every bounded $B$. For $\alpha=\beta=0$, bound~(a)
matches the $\Ot(n^{1.5})$ bound of CDXZ while allowing any constant number of interleaved monotone parts, and in general the bound is $\Ot(n^{1.5}\sqrt{\mono(A)\mono(B)})$. Bound~(a) also holds for sequences
with infinite entries (\cref{thm:nearly-mono-1}), provided that the finite entries can be partitioned into few
monotone subsequences. When the finite entries of each sequence are monotone, the infinite entries may be placed
arbitrarily and the expected running time is $\Ot(n^{1.5})$; this covers both bounded monotone
$(\min,+)$-convolution and minimum-witness convolution. Bound~(b) is of interest when only $A$ is structured: for a
constant number of interleaved monotone subsequences in $A$, it gives $\Ot(n^{5/3})$ time for arbitrary bounded $B$,
while for monotone $A$ an $\Ot(n^{1.5})$ bound is already known (\cref{rem:one-part}). A simple reduction
(\cref{prop:one-sided-lb}) shows that an $O(n^{1.5-\eps})$-time algorithm for $\mono(A)\le 2$ and arbitrary bounded
$B$ would give an $O(n^{2-\eps})$-time algorithm for bounded $(\min,+)$-convolution.

Computing the monotonicity number exactly is NP-hard~\cite{Wagner84}, and the $1.71$-approximation of Fomin,
Kratsch, and Novelli~\cite{FominKN02} (stated for sequences without repeated values) is too slow for our purposes.
A simple greedy algorithm, however, finds a partition into $O(\mono(A)\log n)$ monotone subsequences in
$\Ot(n\cdot\mono(A))$ time (\cref{lem:greedy-partition}); the extra logarithmic factor in the number of parts is
absorbed by our bounds.

\emph{One bounded sequence.}
The standard $\Ot(n(M+1))$ bound for two sequences with entries in $\set{0,\dots,M}$ also holds when only $A$ has
such entries.

\begin{theorem}[see \cref{thm:one-bounded}]
\label{thm:intro-one-bounded}
Let $A\in\set{0,\dots,M}^n$ and let $B\in(\Zinf)^n$ be arbitrary. Then $A\mc B$ can be computed
deterministically in $\Ot(n(M+1))$ time.
\end{theorem}

The proof uses the prefix minima of $B$ to estimate each of the first $n$ output entries up to an additive error of
$M$, and then splits the instance into subinstances in which both sequences have a small value range; this is in the
spirit of the convolution-and-prediction technique of Bateni et al.~\cite{BateniHSS18}. The argument uses that every
entry of $A$ is finite. Once $A$ may contain $+\infty$, we cannot expect $\Ot(n(M+1))$ time: already for
$A\in\set{0,+\infty}^n$ the problem contains \emph{minimum-witness convolution}, where for sets $S,T\subseteq[n]$ one
asks, for every $k$, for the smallest $j\in T$ with $k-j\in S$. The best algorithms for this problem, and for the
$(\min,+)$-convolution of an arbitrary sequence with any $\set{0,+\infty}$-valued mask, take $\Ot(n^{1.5})$
time~\cite{BabaiF09,ChanH20}, and no $O(n^{1.5-\eps})$-time algorithm is known even for minimum-witness
convolution~\cite{LincolnPW20}. Allowing finite entries of $A$ in $\set{0,\dots,M}$, in addition to $+\infty$, costs
only a factor of $\sqrt M$ over this mask case.

\begin{theorem}[see \cref{thm:bounded-inf}]
\label{thm:intro-bounded-inf}
Let $A\in(\set{0,\dots,M}\cup\set{+\infty})^n$ and let $B\in(\Zinf)^n$ be arbitrary, where $1\le M\le n$.
Then $A\mc B$ can be computed deterministically in $\Ot(n^{1.5}\sqrt M)$ time.
\end{theorem}

For comparison, if $A$ takes $h$ distinct finite values, splitting it into $h$ masks gives
$\Ot(hn^{1.5})$ time~\cite{BabaiF09}; \cref{thm:intro-bounded-inf} is faster when $h$ is larger than $\sqrt M$.

\emph{Reductions.}
In the other direction, we reduce bounded $(\min,+)$-convolution to bounded monotone $(\min,+)$-convolution. The
simplest way to make a sequence monotone is to add a steep linear function, which preserves the convolution up to a
known additive term but increases the largest entry to $\Theta(n^2)$. We limit this growth in the values by
increasing the sequence length.

\begin{theorem}[see \cref{thm:mono-red}]
\label{thm:intro-mono}
Let $A,B\in\set{0,\dots,U}^n$ and let $w\ge 1$ be an integer. In $O(nw)$ time one can construct
non-decreasing sequences $A',B'$ of length $N=2nw$ with entries in $\set{0,\dots,n(2\floor{U/w}+3)}$,
such that $A\mc B$ can be computed from $A'\mc B'$ in $O(nw)$ time.
\end{theorem}

With $U=n$ and $w=\ceil{\sqrt n}$, both the length and the value bound of the monotone instance are $\Theta(n^{1.5})$.

\begin{corollary}[see \cref{cor:four-thirds}]
\label{cor:intro-four-thirds}
If bounded monotone $(\min,+)$-convolution of length $N$ can be computed in time $O(N^{4/3-\eps})$ for some
$\eps>0$, then bounded $(\min,+)$-convolution of length $n$ can be computed in time $O(n^{2-3\eps/2})$.
\end{corollary}

A reduction with threshold $3/2$ instead of $4/3$ would show that any polynomial improvement of the CDXZ algorithm
transfers to bounded inputs. Other choices of $w$ give a trade-off (\cref{cor:tradeoff}): for every constant
$\gamma\in[0,1)$, an $O(n^{2-\eps})$-time algorithm for monotone instances of length $\Theta(n^{1+\gamma})$ with
entries bounded by $O(n^{2-\gamma})$ would give an $O(n^{2-\eps}+n^{1+\gamma})$-time algorithm for bounded instances
of length $n$. A similar trade-off between length and value range applies to bounded inputs without monotonicity.

\begin{theorem}[see \cref{thm:value-red}]
\label{thm:intro-value}
Let $A,B\in\set{0,\dots,U}^n$ and let $1\le L\le U$ be an integer. In $O(nU/L)$ time one can construct
sequences $A',B'$ of length $N=2n(\floor{U/L}+1)$ with entries in $\set{0,\dots,2L}$, such that $A\mc B$
can be computed from $A'\mc B'$ in $O(N)$ time.
\end{theorem}

\begin{corollary}[see \cref{cor:value-lb}]
\label{cor:intro-value}
Let $x\in(0,1]$ be a constant. If $(\min,+)$-convolution with entries in $\set{0,\dots,N^{x}}$ can be computed in
time $O(N^{1+x-\eps})$ for some $\eps>0$, then bounded $(\min,+)$-convolution of length $n$ can be computed in time
$O(n^{2-2\eps/(1+x)})$.
\end{corollary}

In particular, a polynomial improvement of the $\Ot(n(M+1))$ bound of \cref{thm:intro-one-bounded} for $M=n^x$, even
when $B$ also has entries in $\set{0,\dots,M}$, would give a polynomial improvement for bounded inputs.

\emph{Historical perspective.}
Chan, Vassilevska Williams, and Xu~\cite{ChanVX23} proposed the Strong Min-Plus Convolution Hypothesis, which
postulates that bounded $(\min,+)$-convolution requires $n^{2-o(1)}$ time. Before~\cite{AlmanV26}, our reductions
could therefore be read as conditional lower bounds: $N^{4/3-o(1)}$ for bounded monotone instances
(\cref{cor:intro-four-thirds}), $N^{1+x-o(1)}$ for entries bounded by $N^x$ (\cref{cor:intro-value}), and
$n^{1.5-o(1)}$ for $\mono(A)\le2$ and arbitrary bounded $B$ (\cref{prop:one-sided-lb}). Since~\cite{AlmanV26} together
with the reduction of Bremner et al.\ gives a truly subquadratic algorithm for bounded $(\min,+)$-convolution, we
state these results only as reductions, which relate the exponents of the different variants.

\subsection{Technical Overview}
\label{sec:overview}

\emph{Barely monotone sequences.}
Both algorithms build on the framework of CDXZ~\cite{ChiDXZ22}. We write $A[i]=D\widetilde A[i]+a_i$ with
$0\le a_i<D$, and similarly for $B$. The coarse convolution $\widetilde C=\widetilde A\mc\widetilde B$ restricts the
coarse value of an optimal pair for $C[k]$ to $\widetilde C[k]$ or $\widetilde C[k]+1$. In the basic form of the
framework, which we use for bound~(b), one multiplies polynomials that record the coarse values only modulo a random prime $p$, keeping the FFT small,
and then subtracts the erroneous terms whose coarse values agree with the target only modulo $p$. For monotone sequences, the
indices with the same coarse value form intervals, and erroneous terms can be handled interval by interval. If $A$
is a union of several monotone subsequences, the indices of one subsequence with the same coarse value only form an
interval of that subsequence. For bound~(b) we process each such group of $A$ against each entry of $B$ with a
simple reporting structure; since the groups of one subsequence have disjoint spans, the false positives caused by
indices of other subsequences can be charged to the random prime.

For bound~(a) we use the bit-by-bit refinement of CDXZ, taking the random prime $p$ itself as the splitting scale
$D$, so that the coarse sums are computed exactly. Erroneous terms are tracked as \emph{segments}:
maximal intervals of indices $i$, for a fixed output $k$, on which all rounded values of $A[i]$ and $B[k-i]$ are
constant. A segment whose coarse sum differs from the target is retained near the current minimum only if a fixed integer
associated with it is close to a nonzero multiple of $p$, which is unlikely. With interleaved parts, we define segments separately for each pair of parts. To prevent indices of other parts
from breaking a segment, we fill the gaps in each part to obtain a \emph{filled} sequence that is monotone on all
of $[n]$. Subtracting a segment then requires counting the pairs $(i,k-i)$ that belong to the two parts,
a range-restricted convolution count. We avoid computing most of these counts with an
\emph{aliasing filter}: at every bit level, the polynomial products also record the coarse values modulo a fresh
random prime $q_\lambda$ (or $q_\lambda=1$ when no filtering is needed), and only erroneous segments that agree with the target modulo $q_\lambda$ need to be
counted, which we do by scanning them. Larger $q_\lambda$ means larger polynomials but less scanning, and balancing
the two at every level gives $\Ot(n^{1.5}\sqrt{m_Am_B})$ time for $m_A$ and $m_B$ parts. Nothing in this argument
depends on the values at indices outside the parts, which is why infinite entries can be placed arbitrarily.

\emph{One bounded sequence.}
Let $A\in\set{0,\dots,M}^n$ with $M\ge1$, and let $B_{\min}$ denote the prefix minima of $B$. For $k<n$ we have
$B_{\min}[k]\le C[k]\le B_{\min}[k]+M$: the lower bound holds since $A\ge0$, and the upper bound is witnessed by
the position of the prefix minimum together with an entry of $A$, which is at most $M$. So only entries
$B[j]\le B_{\min}[k]+M$ matter for $C[k]$. For this sketch, assume $B[0]<\infty$. We group the outputs $k$ by
$\floor{(B_{\min}[0]-B_{\min}[k])/M}$. Each group is an interval of indices, and the relevant entries of $B$ lie in
the same index interval or the preceding one, within a value window of width $2M$. Running the FFT on every group takes
$\Ot(n(M+1))$ time in total, and the outputs $k\ge n$ are handled by reversing both sequences.

\emph{Bounded entries with infinities.}
We partition the finite entries of $B$ into value buckets of width $M$. For each output $k$, the first bucket
containing some $B[j]$ with $A[k-j]<\infty$ determines $C[k]$ up to an additive $2M$; only this bucket and the next
can contain the optimum. To find the first bucket, we follow the approach for mask convolutions~\cite{ChanH20}.
We group the buckets into $O(n/\tau)$ blocks, with each large bucket forming its own block. One FFT per block
identifies the first relevant block for each output; if that block is small, we scan it. The two relevant buckets are evaluated by brute force if they are small and with a precomputed bounded
convolution if they are large. With $\tau=\sqrt{nM}$ the total time is $\Ot(n^{1.5}\sqrt M)$.

\emph{From bounded to monotone.}
Adding $i(U+1)$ to $A[i]$ makes $A$ non-decreasing, but the largest
entry can grow to order $n(U+1)$. Instead, we replace
each entry $x=A[i]$ by a block of $2w$ entries $iK+\floor{(x+j)/w}$, $j=0,\dots,2w-1$. Each block is non-decreasing,
and since the rounded values are $O(1+U/w)$, an offset of $K=\Theta(1+U/w)$ between consecutive blocks is
enough to make the whole sequence monotone, with entries of order $n(1+U/w)$. The rounding loses information, but it
can be recovered from the convolution: the minimum of $\floor{(x+j)/w}+\floor{(y+s-j)/w}$ over any $w$ consecutive
values of $j$ equals $\floor{(x+y+s+1)/w}-1$ and hence depends only on $x+y$ (\cref{lem:floor}). Reading off the
appropriate entries of $A'\mc B'$, we learn $\floor{(C[k]+u)/w}$ for every $u\in\set{1,\dots,w}$, and these values
determine $C[k]$.

\emph{Reducing the value range.}
To go from entries bounded by $U$ to entries bounded by $O(L)$, we split each sequence into $\floor{U/L}+1$ layers.
Layer $c$ keeps the entries in $[cL,(c+1)L)$, shifted down by $cL$, and replaces every other entry by the dummy
value $2L$. We pad the layers to length $2n$, again with the value $2L$, and concatenate them. At the output positions $2nd+k$ we then see the
best contribution of all layer pairs $(c,c')$ with $c+c'=d$; a value below $2L$ can only come from two genuine
entries, so after adding back $dL$ and minimizing over $d$ we recover $C[k]$.

\subsection{Related Work}

Cygan, Mucha, W\k{e}grzycki, and W\l{}odarczyk~\cite{CyganMWW19} and K\"unnemann, Paturi, and
Schneider~\cite{KunnemannPS17} established fine-grained equivalences between $(\min,+)$-convolution and a number of
other problems. Lincoln, Polak, and Vassilevska Williams~\cite{LincolnPW20} studied a family of intermediate matrix
products and their convolution analogues; for the latter, which include minimum-witness convolution,
$\Ot(n^{1.5})$-time algorithms are known but no $O(n^{1.5-\eps})$-time algorithms. Chan, Vassilevska Williams, and Xu~\cite{ChanVX23} studied the fine-grained complexity of a range of intermediate
matrix and convolution problems. Chan and Lewenstein~\cite{ChanL15} connected binary jumbled indexing to bounded monotone
$(\min,+)$-convolution. Amir et al.~\cite{AmirCLL14} studied the hardness of jumbled indexing over alphabets of any fixed size $r\ge3$,
relating it to 3SUM. For the
matrix case, the bounded-difference and monotone instances mentioned above arise in algorithms for language edit
distance and RNA folding~\cite{BringmannGSW16}, and for range mode and single-source replacement
paths~\cite{GuPWX21}. For a survey of fine-grained
complexity, see~\cite{VassilevskaWilliams18}.

\emph{Organization.}
After the preliminaries in \cref{sec:prelims}, we give our algorithms for barely monotone sequences in
\cref{sec:nearly}. \Cref{sec:mono-red} gives the reduction from bounded to bounded monotone $(\min,+)$-convolution,
and \cref{sec:values} treats value ranges and one bounded sequence, with and without infinite entries.
\Cref{sec:history} describes the history of this work.

\section{Preliminaries}
\label{sec:prelims}

We fix notation and the model of computation, and recall two standard tools: bounded $(\min,+)$-convolution via
the FFT, and a bound on the probability that a random prime divides a fixed nonzero integer.

For an integer $n\ge 1$ we write $[n]=\set{0,1,\dots,n-1}$, and for integers $a\le b$, an interval $[a,b]$ used as a
set of indices denotes $\set{a,a+1,\dots,b}$. We write $\mathbf 1[\mathcal E]$ for the indicator of a condition
$\mathcal E$, and $\min\emptyset=+\infty$.
Sequences are indexed from $0$, and entries are in $\Zinf$ unless stated otherwise.
For sequences $A$ of length $n$ and $B$ of length $m$, their $(\min,+)$-convolution is the sequence
$A\mc B$ of length $n+m-1$ with
\[
(A\mc B)[k]=\min\set{A[i]+B[j] \;:\; i\in[n],\, j\in[m],\, i+j=k},
\]
with the convention $x+\infty=\infty$.
We say that a pair $(i,j)$ is a \emph{witness} for output $k$ if $i+j=k$ and $A[i]+B[j]=(A\mc B)[k]<\infty$.
A sequence is \emph{non-decreasing} if $A[i]\le A[i+1]$ for all $i$, \emph{non-increasing} if
$A[i]\ge A[i+1]$ for all $i$, and \emph{monotone} if it is either.
For a set $S$ of indices, a sequence is \emph{monotone on $S$} if its restriction to $S$ (in increasing order of
indices) is monotone.
The notation $\Ot(\cdot)$ suppresses factors polylogarithmic in $n$, and in running-time bounds, $\log n$ stands
for $\log(n+2)$.

We work in the word RAM model with words of $\Theta(\log n)$ bits.
All finite entries are integers bounded in absolute value by a polynomial in $n$, so that they fit into $O(1)$
words. In \cref{thm:one-bounded,thm:bounded-inf}, only the value range of $A$ is bounded by $M$; the finite entries of $B$
may range over a polynomially large interval.
Randomized algorithms are Las Vegas: they always return the correct answer, and we bound their expected
running time.

The following standard FFT lemma uses the technique of Alon, Galil, and Margalit~\cite{AlonGM97};
see also~\cite{CyganMWW19,BringmannDP24}.

\begin{lemma}[Bounded $(\min,+)$-convolution]
\label{lem:fft}
Let $A\in(\Zinf)^{n}$ and $B\in(\Zinf)^{m}$, and suppose that there are integers $a,b$ and $M\ge 0$ such that
every finite entry of $A$ lies in $[a,a+M]$ and every finite entry of $B$ lies in $[b,b+M]$.
Then $A\mc B$ can be computed deterministically in $O((n+m)(M+1)\log(n+m+M))$ time.
\end{lemma}

\begin{proof}
Consider the bivariate polynomials
\[
P_A(x,z)=\sum_{i:\,A[i]<\infty} x^{A[i]-a}z^{i}\qquad\text{and}\qquad P_B(x,z)=\sum_{j:\,B[j]<\infty}x^{B[j]-b}z^{j}.
\]
The coefficient of $x^{e}z^{k}$ in $P_AP_B$ is the number of pairs $(i,j)$ with $i+j=k$ and
$A[i]+B[j]=a+b+e$. Hence $(A\mc B)[k]=a+b+\min\set{e: [x^ez^k]P_AP_B>0}$ (and $+\infty$ if no such $e$
exists).
All coefficients are nonnegative integers bounded by $\min(n,m)$.
Substituting $z=x^{2M+1}$ yields univariate polynomials of degree $O((n+m)(M+1))$ from which the
coefficients of $P_AP_B$ can be read off unambiguously (since $0\le e\le 2M$, the pair $(e,k)$ is determined by $e+(2M+1)k$), and the product can be computed exactly with the
FFT over a suitable prime field (or via fast integer multiplication) in the stated time.
\end{proof}

\begin{lemma}
\label{lem:prime}
There is an absolute constant $c_0$ such that the following holds.
Let $\mathsf{P}\ge 2$. The interval $[\mathsf{P},2\mathsf{P}]$ contains at least $\mathsf{P}/(c_0\ln(2\mathsf{P}))$
primes, and if $p$ is chosen uniformly at random among them, then $\Pr[p\mid z]\le c_0\ln(|z|+1)/\mathsf{P}$ for
every nonzero integer $z$.
Moreover, $p$ can be sampled in $\Ot(\mathsf{P})$ time.
\end{lemma}

\begin{proof}
By Chebyshev-type estimates (or the prime number theorem), the interval $[\mathsf{P},2\mathsf{P}]$ contains at least
$\mathsf{P}/(c_1\ln(2\mathsf{P}))$ primes for an absolute constant $c_1$; we take $c_0=2c_1$.
A nonzero integer $z$ has at most $\log_{\mathsf{P}}|z|$ distinct prime factors that are at least $\mathsf{P}$.
Hence, using $\ln(2\mathsf{P})\le2\ln\mathsf{P}$, we get $\Pr[p\mid z]\le \log_{\mathsf{P}}|z|\cdot c_1\ln(2\mathsf{P})/\mathsf{P}\le 2c_1\ln|z|/\mathsf{P}\le c_0\ln(|z|+1)/\mathsf{P}$.
The primes in $[\mathsf{P},2\mathsf{P}]$ can be listed with the sieve of Eratosthenes in $\Ot(\mathsf{P})$ time.
\end{proof}

\section{Barely Monotone Sequences}
\label{sec:nearly}

Both algorithms for barely monotone sequences combine monotone partitions with the framework of CDXZ. We first
prove bound~(b) of \cref{thm:intro-nearly-mono}, which needs a partition of $A$ only (\cref{sec:nearly-2}), and then the
more involved bound~(a) (\cref{sec:nearly-1}).

\subsection{Monotone Partitions}
\label{sec:partitions}

A \emph{monotone partition} of a sequence $A$ of length $n$ is a partition $[n]=P_1\cup\dots\cup P_m$ of its
index set such that $A$ is monotone on every $P_s$. Thus $\mono(A)$ is the minimum size of a monotone partition.

\begin{lemma}
\label{lem:greedy-partition}
There is a deterministic algorithm that, given a sequence $A$ of length $n$, computes a monotone partition of $A$
into at most $\min\set{\mono(A)\cdot(\ln n+1),\ 3\sqrt n+\log_2 n+1}$ parts in
$O(\min\set{\mono(A)\log n,\sqrt n}\cdot n\log n)$ time. In particular, $\mono(A)=O(\sqrt n)$.
\end{lemma}

\begin{proof}
The algorithm maintains the set $Y\subseteq[n]$ of indices not yet assigned, initially $Y=[n]$.
While $Y\neq\emptyset$, it computes a longest subsequence of $A|_Y$ that is non-decreasing and a longest one that
is non-increasing (each in $O(|Y|\log|Y|)$ time~\cite{Fredman75}), removes the longer of the two from $Y$, and
makes it a part.

Let $m=\mono(A)$. Restricting an optimal monotone partition to $Y$ shows that $A|_Y$ can be partitioned into at
most $m$ monotone subsequences, so one of them has length at least $|Y|/m$. Hence each iteration leaves at most a fraction
$1-1/m$ of the remaining indices. If $m=1$, one iteration suffices. Otherwise, after $t$ iterations
$|Y|\le n(1-1/m)^t\le ne^{-t/m}$, which is less than $1$ for $t>m\ln n$. Thus there are at most $m\ln n+1$
iterations.
Moreover, by the Erd\H{o}s--Szekeres theorem~\cite{ErdosS35}, every sequence of length $r$ contains a monotone
subsequence of length at least $\sqrt r$. While $|Y|>r_0$ for a threshold $r_0$, each iteration removes at least
$\sqrt{r_0}$ elements. Hence, for $h=0,1,2,\dots$, the number of iterations that start with
$n/4^{h+1}<|Y|\le n/4^{h}$ is at most $(3n/4^{h+1})/\sqrt{n/4^{h+1}}+1=\tfrac32\sqrt{n/4^{h}}+1$. Summing over the
at most $\log_4n+1$ relevant values of $h$ gives at most $3\sqrt n+\log_2n+1$ iterations.
Each iteration takes $O(n\log n)$ time.
\end{proof}

By \cref{lem:greedy-partition}, both algorithms can afford to compute a partition of $A$ into
$O(\min\set{\mono(A)\log n,\allowbreak\sqrt n})$ parts and, for bound~(a), of $B$ into $O(\min\set{\mono(B)\log n,\sqrt n})$
parts; since their running times depend polynomially on the number of parts, this loses only polylogarithmic
factors.

\begin{remark}
\label{rem:min-witness}
The $(\min,+)$-convolution of two bounded sequences with monotonicity number at most~$2$ already contains
minimum-witness convolution (see
\cref{rem:min-witness-conv}): given $S,T\subseteq[n]$, let $A[i]=0$ for $i\in S$ and $A[i]=2n$ otherwise, and let
$B[j]=j$ for $j\in T$ and $B[j]=2n$ otherwise.
Then $A$ is constant on $S$ and on $[n]\setminus S$, $B$ is increasing on $T$ and constant on $[n]\setminus T$,
and $(A\mc B)[k]<2n$ if and only if some $j\in T$ has $k-j\in S$, in which case $(A\mc B)[k]$ is the smallest
such $j$.
\end{remark}

\subsection{The Coarse-Value Framework}
\label{sec:framework}

We use the framework of CDXZ~\cite{ChiDXZ22} in the following form.
Let $A,B\in\set{0,\dots,U}^n$ with $U=O(n)$, write $C=A\mc B$, and let $D\ge1$ be an integer.
Write
\[
A[i]=D\widetilde A[i]+a_i,\qquad B[j]=D\widetilde B[j]+b_j,\qquad 0\le a_i,b_j<D,
\]
and let $\widetilde C=\widetilde A\mc\widetilde B$ be the \emph{coarse convolution}.

\begin{lemma}
\label{lem:coarse}
For every output $k$ and every witness $(i,j)$ of $C[k]$ we have
$\widetilde A[i]+\widetilde B[j]\in\set{\widetilde C[k],\widetilde C[k]+1}$.
\end{lemma}

\begin{proof}
Every pair $(i',j')$ with $i'+j'=k$ has $\widetilde A[i']+\widetilde B[j']\ge\widetilde C[k]$.
A pair $(i',j')$ attaining $\widetilde C[k]$ has $A[i']+B[j']\le D\widetilde C[k]+2D-2$, so
$C[k]\le D\widetilde C[k]+2D-2$.
A pair with $\widetilde A[i]+\widetilde B[j]\ge\widetilde C[k]+2$ has value at least
$D\widetilde C[k]+2D>C[k]$, so it is not a witness.
\end{proof}

Let $p$ be a prime, and consider the trivariate polynomials
\[
F_A(x,y,z)=\sum_{i\in[n]}x^{a_i}\,y^{\widetilde A[i]\bmod p}\,z^{i},\qquad
F_B(x,y,z)=\sum_{j\in[n]}x^{b_j}\,y^{\widetilde B[j]\bmod p}\,z^{j}.
\]
For an output $k$ and an offset $\delta\in\set{0,1}$, let $t=\widetilde C[k]+\delta$ be the \emph{target}, and
let
\[
Q_{k,\delta}(x)=\sum_{\substack{0\le e\le 2p-2\\ e\equiv t\ (\mathrm{mod}\ p)}}[y^ez^k]\,(F_AF_B)
=\sum_{\substack{i+j=k\\ \widetilde A[i]+\widetilde B[j]\equiv t\ (\mathrm{mod}\ p)}}x^{a_i+b_j}.
\]
Call a pair $(i,j)$ with $i+j=k$ \emph{erroneous for $(k,\delta)$} if
$\widetilde A[i]+\widetilde B[j]\equiv t\pmod p$ but $\widetilde A[i]+\widetilde B[j]\neq t$, and let
$R_{k,\delta}(x)=\sum x^{a_i+b_j}$, where the sum ranges over the pairs that are erroneous for $(k,\delta)$.

\begin{lemma}
\label{lem:framework}
For every $k$, if $C[k]<\infty$ then
\[
C[k]=\min_{\delta\in\set{0,1}}\Bigl(D(\widetilde C[k]+\delta)+\min\set{e:[x^e](Q_{k,\delta}-R_{k,\delta})>0}\Bigr),
\]
where the inner minimum of an empty set is $+\infty$.
Given $\widetilde C$ and all polynomials $R_{k,\delta}$, the sequence $C$ can be computed in
$\Ot(nDp)$ time.
\end{lemma}

\begin{proof}
By definition, $Q_{k,\delta}-R_{k,\delta}=\sum x^{a_i+b_j}$, where the sum ranges over the pairs $(i,j)$ with
$i+j=k$ and $\widetilde A[i]+\widetilde B[j]=\widetilde C[k]+\delta$, all of which have value
$D(\widetilde C[k]+\delta)+a_i+b_j$. So each term of the outer minimum is the minimum value of such a pair, and
by \cref{lem:coarse} some witness of $C[k]$ is among these pairs for some $\delta\in\set{0,1}$.
The polynomial $F_AF_B$ has $x$-degree at most $2D-2$, $y$-degree at most $2p-2$ and $z$-degree at most $2n-2$, and
its coefficients are nonnegative integers at most $n$. By Kronecker substitution it can be computed exactly with
the FFT in $\Ot(nDp)$ time, after which every $Q_{k,\delta}$ can be read off in $O(D)$ time.
\end{proof}

It therefore suffices to compute $\widetilde C$ and to enumerate the erroneous pairs. If $p$ is chosen at random, a
fixed pair is erroneous only with small probability.

\begin{lemma}
\label{lem:erroneous-prob}
Let $p$ be a uniformly random prime in $[\mathsf{P},2\mathsf{P}]$. For every fixed $k$, $\delta$ and integer $v\neq\widetilde C[k]+\delta$
with $|v|\le n^{O(1)}$, we have $\Pr[v\equiv \widetilde C[k]+\delta\pmod p]=O(\log n/\mathsf{P})$.
\end{lemma}

\begin{proof}
Apply \cref{lem:prime} to $z=v-\widetilde C[k]-\delta\neq 0$.
\end{proof}

To enumerate erroneous pairs, we use the following reporting structure.

\begin{lemma}
\label{lem:reporting}
Let $X\in\Z^{N}$, let $p\ge1$ be an integer, and let $\delta\in\Z$. After $O(N+p)$ preprocessing time, the following queries can
be answered in $O(\log N)$ time plus $O(1)$ time per reported position: given an interval $[\ell,r]$ and an integer $\Delta$, report all
$\kappa\in[\ell,r]$ with $X[\kappa]+\delta\equiv\Delta\pmod p$ and $X[\kappa]+\delta\neq\Delta$.
\end{lemma}

\begin{proof}
For every residue $\rho\in[p]$, store the sorted list $L_\rho$ of all positions $\kappa$ with
$X[\kappa]+\delta\equiv\rho\pmod p$, and for every entry of $L_\rho$ a pointer $\mathrm{nxt}$ to the next entry
of $L_\rho$ whose $X$-value differs from its own (or to the end of the list).
To answer a query, binary search for the first entry of $L_{\Delta\bmod p}$ that is at least $\ell$, and walk
along the list while the current position is at most $r$: if the current position $\kappa$ satisfies
$X[\kappa]+\delta=\Delta$, follow $\mathrm{nxt}$, which skips only entries of value $X[\kappa]$; otherwise report
$\kappa$ and move to the next entry.
Every step either reports a position or is a jump that is immediately followed by a reported position or by the
end of the walk, so the walk takes at most $2z+1$ steps, where $z$ is the number of reported positions.
\end{proof}

Let $P_1,\dots,P_{m}$ be a monotone partition of $A$. For each part $P_s$, the indices $i\in P_s$ with a fixed
value of $\widetilde A[i]$ form a contiguous range of $P_s$ (in increasing index order), because $A$, and hence
$\widetilde A$, is monotone on $P_s$. We call these ranges the \emph{groups} of $A$ (at scale $D$); for a group
$G\subseteq P_s$ we write $\widetilde A[G]$ for its common coarse value and $I_G=[\min G,\max G]$ for its
\emph{span}. The groups have the following properties:
\begin{enumerate}[label=(G\arabic*)]
\item\label{it:g1} Every index belongs to exactly one group, and $G=P_s\cap I_G$ for every group $G\subseteq P_s$.
\item\label{it:g2} The spans of the groups of a fixed part $P_s$ are pairwise disjoint; hence the sum of the span
lengths of all groups is at most $mn$.
\item\label{it:g3} Each part has at most $\floor{U/D}+1$ groups, so there are at most $m(U/D+1)$ groups in total.
\end{enumerate}

\subsection{An Algorithm Using Only One Monotone Partition}
\label{sec:nearly-2}

\begin{theorem}
\label{thm:nearly-mono-2}
Let $A,B\in\set{0,\dots,U}^n$ with $U=O(n)$, and let $m=\mono(A)$. Then $A\mc B$ can be computed by a Las Vegas
algorithm in expected time $\Ot(n^{5/3}m^{2/3})$.
\end{theorem}

\begin{proof}
Using \cref{lem:greedy-partition}, we compute a monotone partition of $A$ into parts
$P_1,\dots,P_{m'}$, where $m'=O(\min\set{m\log n,\sqrt n})$, in $\Ot(n\min\set{m,\sqrt n})$ time; this is within the claimed bound.
The multiplication costs $\Ot(nD\mathsf{P})$, the queries $\Ot(m'n^2/D)$, and processing the reports an expected
$\Ot(m'n^2/\mathsf{P})$; we choose $D$ and $\mathsf{P}$ to equalize these.
Let $D=\ceil{(m'n)^{1/3}}$ and $\mathsf{P}=\max\set{2,\ceil{(m'n)^{1/3}}}$, and let $p$ be a uniformly random prime in $[\mathsf{P},2\mathsf{P}]$.

The sequences $\widetilde A,\widetilde B$ have entries in $\set{0,\dots,\floor{U/D}}$, so by \cref{lem:fft},
$\widetilde C$ can be computed in $\Ot(n(1+U/D))=\Ot(n^2/D)$ time.

For each $\delta\in\set{0,1}$, build the structure of \cref{lem:reporting} on $X=\widetilde C$ with
shift $\delta$.
For every group $G$ of $A$, say $G\subseteq P_s$, every $j\in[n]$, and every $\delta\in\set{0,1}$, query the
interval $I_G+j$ with $\Delta=\widetilde A[G]+\widetilde B[j]$. For every reported output $\kappa$, let
$i=\kappa-j\in I_G$; if $i\in P_s$, add the term $x^{a_i+b_j}$ to $R_{\kappa,\delta}$, and otherwise discard
$\kappa$ (a \emph{false report}).

To see that this enumerates exactly the erroneous pairs, let $(i,j)$ be erroneous for $(k,\delta)$, and let $G$ be
the group containing $i$, say $G\subseteq P_s$. Then
$k=i+j\in I_G+j$, and since $\widetilde A[i]=\widetilde A[G]$, the output $k$ is reported by the query for
$(G,j,\delta)$; as $i\in P_s$ the term is added. Conversely, every added term corresponds to a pair $(i,j)$ with
$i\in P_s\cap I_G=G$ by~\ref{it:g1}, so $\widetilde A[i]+\widetilde B[j]=\Delta$, and the query guarantees that
$(i,j)$ is erroneous for $(i+j,\delta)$. Each pair $(i,j)$ is considered only in the query for the unique group
containing $i$, so every erroneous pair is added exactly once. By \cref{lem:framework}, the algorithm is correct
for every choice of $p$.

The polynomial multiplication takes $\Ot(nD\mathsf{P})$ time by \cref{lem:framework}.
By~\ref{it:g3} the number of queries is $O(n\cdot m'(U/D+1))=O(m'n^2/D)$, each with an
$O(\log n)$ overhead.
A query for $(G,j,\delta)$ reports $\kappa\in I_G+j$, possibly as a false report, only if $\widetilde A[G]+\widetilde B[j]\neq\widetilde C[\kappa]+\delta$
but the two are congruent modulo $p$. All quantities involved are determined before $p$ is sampled and bounded by
$O(n)$, so by \cref{lem:erroneous-prob} this happens with probability $O(\log n/\mathsf{P})$. The number of triples $(G,j,\kappa)$ with $\kappa\in I_G+j$ is $n\sum_G|I_G|\le m'n^2$ by~\ref{it:g2}, so by
linearity of expectation (over these triples and the two choices of $\delta$) the expected number of reports is
$O(m'n^2\log n/\mathsf{P})$.
In total, the expected running time is
\[
\Ot\Bigl(\frac{n^2}{D}+nD\mathsf{P}+\frac{m'n^2}{D}+\frac{m'n^2}{\mathsf{P}}\Bigr)=\Ot\bigl(n^{5/3}m'^{2/3}\bigr)
=\Ot\bigl(n^{5/3}m^{2/3}\bigr). \qedhere
\]
\end{proof}

\begin{remark}
\label{rem:one-part}
The algorithm of \cref{thm:nearly-mono-2} requires only that the entries of $B$ be bounded by $O(n)$.
For $m=1$, the one-monotone reduction of~\cite[Theorem~8]{BringmannDP24} and the bounded monotone algorithm
of~\cite{ChiDXZ22,BringmannDP24} already give the faster expected time $\Ot(n^{1.5})$.
Our result applies when $A$ has several interleaved monotone subsequences: for $m=O(1)$,
\cref{thm:nearly-mono-2} computes the $(\min,+)$-convolution with an arbitrary bounded sequence in $\Ot(n^{5/3})$ time.
\end{remark}

For every fixed $m\ge2$, improving the exponent $5/3$ to below $3/2$ would also improve bounded
$(\min,+)$-convolution, by the following reduction.

\begin{proposition}
\label{prop:one-sided-lb}
Let $2\le m\le n$, and suppose that $(\min,+)$-convolution of sequences $A,B\in\set{0,\dots,2n+1}^n$ with $\mono(A)\le m$ and no restriction on the monotonicity number of $B$ can be solved in time $T(n,m)$. Then bounded $(\min,+)$-convolution of length $n$ can be solved in time
$O\bigl((1+\sqrt n/(m-1))\cdot T(n,m)+n^{1.5}\log n\bigr)$. In particular, if $T(n,m)=O(m\,n^{1.5-\eps})$ for some
$\eps>0$ and some $m=m(n)$ with $2\le m\le\sqrt n$, then bounded $(\min,+)$-convolution can be solved in
$O(n^{2-\eps}+n^{1.5}\log n)$ time; this holds already for $m=2$ and an $O(n^{1.5-\eps})$-time algorithm.
\end{proposition}

\begin{proof}
Let $A,B\in\set{0,\dots,n}^n$. Compute a monotone partition of $A$ into $r=O(\sqrt n)$ parts with
\cref{lem:greedy-partition}, and group the parts into $\ceil{r/(m-1)}$ groups of at most $m-1$ parts each. For a group
$\mathcal G$ with union $S_{\mathcal G}$, let $A_{\mathcal G}[i]=A[i]$ for $i\in S_{\mathcal G}$ and $A_{\mathcal G}[i]=2n+1$
otherwise. Then $\mono(A_{\mathcal G})\le m$, since the indices outside $S_{\mathcal G}$ carry a constant value. We claim
that $A\mc B$ is the entrywise minimum of the sequences $A_{\mathcal G}\mc B$. Indeed, every pair $(i,j)$ appears with its original value $A[i]+B[j]$ in the instance of the group containing $i$.
Any pair involving the value $2n+1$ has value at least $2n+1$, while every output $k\in[2n-1]$ has an original pair
of value at most $2n$; so pairs involving $2n+1$ never attain the minimum over all instances. The running time is $O(\ceil{r/(m-1)}\cdot T(n,m))$ plus $O(n^{1.5}\log n)$ for the
partition and $O(n\sqrt n)$ for combining the results.
If $T(n,m)=O(m\,n^{1.5-\eps})$ for some $m=m(n)$ with $2\le m\le\sqrt n$, we run the reduction with this $m$; the total time is
$O\bigl((\sqrt n/(m-1)+1)\cdot m\,n^{1.5-\eps}+n^{1.5}\log n\bigr)=O(n^{2-\eps}+n^{1.5}\log n)$, using $m\le\sqrt n$.
\end{proof}

Hence, for a constant number $m\ge2$ of monotone parts in $A$ and arbitrary bounded $B$, an algorithm faster than
$n^{1.5}$ by a polynomial factor would carry over to bounded $(\min,+)$-convolution, while \cref{thm:nearly-mono-2}
gives $\Ot(n^{5/3})$. More generally, for $2\le m\le\sqrt n$, the threshold $m\,n^{1.5}$ of the reduction meets the
upper bound $\Ot(n^{5/3}m^{2/3})$ of \cref{thm:nearly-mono-2} at $m=\Theta(\sqrt n)$.

\subsection{Both Sequences Barely Monotone}
\label{sec:nearly-1}

Bound~(a) of \cref{thm:intro-nearly-mono} extends to sequences with infinite entries, provided that their finite
entries can be partitioned into few monotone subsequences.

\begin{theorem}
\label{thm:nearly-mono-1}
Let $A,B\in(\set{0,\dots,U}\cup\set{+\infty})^n$, where $U=O(n)$. Suppose that the finite entries of $A$ can be
partitioned into $m_A$ monotone subsequences and those of $B$ into $m_B$ monotone subsequences.
Then $A\mc B$ can be computed by a Las Vegas algorithm in expected time $\Ot(n+n^{1.5}\sqrt{m_Am_B})$.
\end{theorem}

For $m_A=m_B=1$, \cref{thm:nearly-mono-1} gives $\Ot(n^{1.5})$ time whenever the finite entries of each sequence
are monotone, with infinite entries placed arbitrarily.
This common generalization of bounded monotone $(\min,+)$-convolution and minimum-witness convolution
(\cref{rem:min-witness-conv}) is not covered by the CDXZ algorithm: CDXZ introduce infinite entries in intermediate instances, but their analysis
relies on a bound on the number of intervals of infinite entries~\cite[Lemma~4.3]{ChiDXZ22}.

The algorithm of \cref{sec:nearly-2} processes each group of $A$ against individual entries of $B$. To exploit
structure in both inputs, we instead refine the values bit by bit, following CDXZ, and process
\emph{segments}: for a fixed output $k$, maximal intervals of indices $i$ on which all relevant rounded values of
$A[i]$ and $B[k-i]$ are constant. At each level, CDXZ subtract erroneous contributions segment by segment, using the number of pairs in each
segment. For monotone sequences, this count is the segment's length.

To keep the gaps in a part from breaking its segments, we extend each part to a monotone sequence on all of $[n]$
and define segments using these \emph{filled} sequences. A segment of the part pair $(P_s,Q_t)$ may contain indices $i$ with $i\notin P_s$ or $k-i\notin Q_t$.
Its length therefore need not equal the number of contributing pairs: for instance, if $k=5$, $P_s=\set{0,2,5}$
and $Q_t=\set{0,1,5}$, and $[0,5]$ is a segment, then it has length six, but only $i=0$ and $i=5$ satisfy $i\in P_s$ and
$k-i\in Q_t$. Computing that number requires a
range-restricted convolution count, for which we know no fast procedure. Instead, we use an
\emph{aliasing filter}: at each bit level, a global polynomial product also records coarse values modulo a fresh
random prime $q_\lambda$ (or $q_\lambda=1$ when no filtering is needed).
True contributions always pass this filter, and a false segment needs to be counted only if it also passes, in
which case we scan it position by position. Increasing $q_\lambda$ enlarges the polynomial product but reduces
the expected scanning cost.

At each level, we multiply two polynomials to count candidate pairs and subtract the contributions of aliased
false segments. We then refine the minimum for each candidate coarse sum and discard candidates that cannot be
optimal, and we retain the false segments close to the current minimum for the next level. The moduli $q_\lambda$ determine which segments are scanned, but the retained lists depend only on $p$. We therefore condition on $p$ when bounding the expected scanning
cost over the moduli $q_\lambda$, and then average over $p$.

\subsubsection{Setup}
\label{sec:a-setup}

Let $A,B$ be as in \cref{thm:nearly-mono-1}, and let $S_A=\set{i: A[i]<\infty}$ and $S_B=\set{j:B[j]<\infty}$.
Let $C=A\mc B$. If $S_A=\emptyset$ or $S_B=\emptyset$, then $C\equiv\infty$, which can be output in $O(n)$ time.
From now on, we assume that both supports are nonempty and that monotone partitions
$S_A=P_1\cup\dots\cup P_{m'_A}$ and $S_B=Q_1\cup\dots\cup Q_{m'_B}$ of the finite entries into nonempty parts are
given, so $m'_A,m'_B\ge1$. If partitions as in the hypothesis are given, then $m'_A\le m_A$ and $m'_B\le m_B$;
otherwise, \cref{lem:greedy-partition} computes partitions with $m'_A=O(m_A\log n)$ and $m'_B=O(m_B\log n)$.
Each part has a direction (non-decreasing or non-increasing). For $i\in S_A$ let $\pi_A(i)$ be the index of
the part containing $i$, and define $\pi_B$ analogously. Let $R=m'_Am'_B$, so that $R=O(m_Am_B\log^2n)$.

For $k\in\mathcal K:=[2n-1]$ let $I_k=\set{i: 0\le i<n,\ 0\le k-i<n}$, an interval. We call
$i\in I_k$ \emph{valid} for $k$ if $i\in S_A$ and $k-i\in S_B$; in this case its \emph{type} is
$(\pi_A(i),\pi_B(k-i))$. Let $V_k$ be the set of valid $i\in I_k$, and $V_k^{s,t}$ the set of valid $i$ of
type $(s,t)$, so that $V_k$ is the disjoint union of the sets $V_k^{s,t}$.
Let $\mathcal K^*=\set{k: V_k\neq\emptyset}$; clearly $C[k]=\infty$ if and only if
$k\notin\mathcal K^*$.

\emph{Parameters.}
In this subsection, the random prime $p$ also serves as the scale at which values are split into high and low parts;
this differs from \cref{sec:framework}, where the scale $D$ and the hashing prime $p$ are separate parameters.
Apart from this, the notation of this subsection is independent of \cref{sec:framework,sec:nearly-2}. The
additional moduli $q_\lambda$ introduced in \cref{sec:a-algo} filter high-part sums at individual bit levels.
The target time is already quadratic when $R$ is a sufficiently large fraction of $n$, so we handle that case
directly; the numerical constants below provide slack for later divisibility estimates.
Let $U_*=\max\set{U,n}=O(n)$. If $n<2^{22}R$, we use the naive $O(n^2)=O(n^{1.5}\sqrt R)$ algorithm; so assume
$n\ge 2^{22}R$, hence $U_*\ge n\ge 2^{22}R$ and $U_*\ge 2^{22}$.
Let $\mathsf{P}=\ceil{\sqrt{RU_*}}$; then $64\le\sqrt{U_*}\le \mathsf{P}\le U_*/1024$.
The polynomial products at levels without filtering take $\Ot(n\mathsf{P})$ time (\cref{lem:a-polycost}) and
tracking segments takes $\Ot(nRU_*/\mathsf{P})$ time (\cref{lem:a-track}); this choice of $\mathsf{P}$ makes the two
equal.
Let $p$ be a uniformly random prime in $[\mathsf{P},2\mathsf{P}]$ and let $h$ be the integer with $2^{h-1}<p\le2^{h}$; thus
$2^h<2p\le 4\mathsf{P}$.

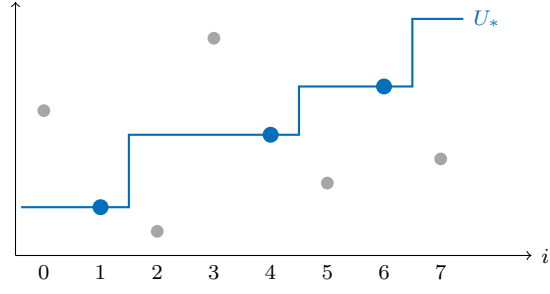
\begin{figure}[t]
\centering
\begin{tikzpicture}[x=0.75cm,y=0.32cm]
  \draw[->] (-0.5,0) -- (8.6,0) node[right,font=\scriptsize] {$i$};
  \draw[->] (-0.5,0) -- (-0.5,10.5);
  \foreach \i in {0,...,7} { \node[below,font=\scriptsize] at (\i,0) {$\i$}; }
  \foreach \i/\v in {0/6,2/1,3/9,5/3,7/4} { \fill[black!35] (\i,\v) circle (2.4pt); }
  \draw[thick,NavyBlue] (-0.4,2) -- (1.5,2) -- (1.5,5) -- (4.5,5) -- (4.5,7) -- (6.5,7) -- (6.5,9.8) -- (7.4,9.8);
  \node[NavyBlue,font=\scriptsize,right] at (7.4,9.8) {$U_*$};
  \foreach \i/\v in {1/2,4/5,6/7} { \fill[NavyBlue] (\i,\v) circle (3pt); }
\end{tikzpicture}
\caption{A non-decreasing part $P_s=\set{1,4,6}$ (dark dots) interleaved with entries of other parts (grey dots), and
its filled sequence $\bar A_s=(2,2,5,5,5,7,7,U_*)$ (step function). Filling lets the segments of $P_s$ ignore the
indices of other parts; contributions are still only counted for indices in $P_s$.}
\label{fig:filled}
\end{figure}

\emph{Filled sequences.}
For every part $P_s$ and $i\in[n]$, if $\set{i'\in P_s: i'\ge i}$ is nonempty, let
$\bar A_s[i]=A[\min\set{i'\in P_s: i'\ge i}]$.
Otherwise, let $\bar A_s[i]=U_*$ if $P_s$ is non-decreasing, and $\bar A_s[i]=0$ if $P_s$ is non-increasing.
Then $\bar A_s$ is monotone on all of $[n]$, in the same direction as $A$ on $P_s$; it has entries in
$[0,U_*]$, and $\bar A_s[i]=A[i]$ for $i\in P_s$. Define $\bar B_t$ analogously.
The filled values only define segments; all counted contributions come from the original finite entries
(see \cref{fig:filled} for an example with $n=8$).
For $k\in\mathcal K^*$, parts $s,t$, and $i\in I_k$, let
\[
\mathrm{val}_{s,t}(k,i)=\bar A_s[i]+\bar B_t[k-i]\in[0,2U_*],\qquad \Delta_{s,t}(k,i)=\mathrm{val}_{s,t}(k,i)-C[k]\in[-2U_*,2U_*].
\]
For $i\in V_k^{s,t}$ we have $\mathrm{val}_{s,t}(k,i)=A[i]+B[k-i]$; in particular
$C[k]=\min_{s,t}\min_{i\in V_k^{s,t}}\mathrm{val}_{s,t}(k,i)$. The quantity $\Delta_{s,t}(k,i)$ does not depend on any random choices and is used only in the analysis.

\emph{High and low parts.}
Let $\widetilde A_s[i]=\floor{\bar A_s[i]/p}$ and $A^*_s[i]=\bar A_s[i]-p\widetilde A_s[i]\in[0,p-1]$, and
define $\widetilde B_t,B^*_t$ analogously. For valid $i$ of type $(s,t)$ we abbreviate
$\widetilde A[i]=\widetilde A_s[i]$, $A^*[i]=A^*_s[i]$, $\widetilde B[k-i]=\widetilde B_t[k-i]$ and
$B^*[k-i]=B^*_t[k-i]$; these depend only on $A[i]$ and $B[k-i]$.
For $k\in\mathcal K^*$ let $H[k]=\min_{i\in V_k}(\widetilde A[i]+\widetilde B[k-i])$, and for $g\in\set{0,1}$ let
\begin{gather*}
T_g(k)=\set{i\in V_k:\ \widetilde A[i]+\widetilde B[k-i]=H[k]+g},\\
\Gamma_g[k]=\min_{i\in T_g(k)}\bigl(A^*[i]+B^*[k-i]\bigr),\qquad
C_g[k]=p(H[k]+g)+\Gamma_g[k],
\end{gather*}
with $\Gamma_g[k]=C_g[k]=\infty$ if $T_g(k)=\emptyset$. The elements of $T_g(k)$ are called \emph{true} for $g$. Thus $C_g[k]$ is the optimum restricted to the valid pairs
whose coarse sum is $H[k]+g$. The algorithm refines these candidate optima without enumerating the sets $T_g(k)$, and discards a candidate once it cannot attain $C[k]$.

\begin{lemma}
\label{lem:a-groups}
For $k\in\mathcal K^*$ we have $T_0(k)\neq\emptyset$, $C_g[k]=\min_{i\in T_g(k)}(A[i]+B[k-i])$, and
$C[k]=\min\set{C_0[k],C_1[k]}$.
\end{lemma}

\begin{proof}
For valid $i$ we have $A[i]+B[k-i]=p(\widetilde A[i]+\widetilde B[k-i])+A^*[i]+B^*[k-i]$ with
$0\le A^*[i]+B^*[k-i]\le 2p-2$, which gives the middle identity. A valid pair attaining $H[k]$ has value less
than $p(H[k]+2)$, while valid pairs with $\widetilde A[i]+\widetilde B[k-i]\ge H[k]+2$ have value at least
$p(H[k]+2)$. So all witnesses of $C[k]$ lie in $T_0(k)\cup T_1(k)$.
\end{proof}

For $\lambda\in[0,h]$ and valid $i$ let $A^{(\lambda)}[i]=\floor{A^*[i]/2^\lambda}$,
$B^{(\lambda)}[k-i]=\floor{B^*[k-i]/2^\lambda}$, and
$S^{(\lambda)}(k,i)=A^{(\lambda)}[i]+B^{(\lambda)}[k-i]$; in particular, $S^{(h)}\equiv0$. The filled versions
$A^{(\lambda)}_s[i]=\floor{A^*_s[i]/2^\lambda}$ and $B^{(\lambda)}_t[j]$ are defined for all indices.
For $\lambda<h$ let $a_\lambda[i]=A^{(\lambda)}[i]-2A^{(\lambda+1)}[i]\in\set{0,1}$ and
$b_\lambda[j]=B^{(\lambda)}[j]-2B^{(\lambda+1)}[j]\in\set{0,1}$, so that
\begin{equation}
\label{eq:a-bits}
S^{(\lambda)}(k,i)=2S^{(\lambda+1)}(k,i)+a_\lambda[i]+b_\lambda[k-i].
\end{equation}
If $T_g(k)\ne\emptyset$, let $\Gamma^{(\lambda)}_g[k]=\min_{i\in T_g(k)}S^{(\lambda)}(k,i)$; thus
$\Gamma^{(0)}_g=\Gamma_g$ and $\Gamma^{(h)}_g=0$.
Since each summand is rounded before minimizing, $\Gamma^{(\lambda)}_g[k]$ need not equal
$\floor{\Gamma_g[k]/2^\lambda}$, and the minimizing pair may differ between levels.

\begin{lemma}
\label{lem:a-levels}
Let $k\in\mathcal K^*$ with $T_g(k)\neq\emptyset$, and $\lambda\in[0,h]$.
\begin{enumerate}[label=(\alph*)]
\item For every valid $i$: $2^\lambda S^{(\lambda)}(k,i)\le A^*[i]+B^*[k-i]\le 2^\lambda S^{(\lambda)}(k,i)+2^{\lambda+1}-2$.
The same holds for the filled quantities of any type $(s,t)$ and any $i\in I_k$.
\item $2^\lambda\Gamma^{(\lambda)}_g[k]\le\Gamma_g[k]\le 2^\lambda\Gamma^{(\lambda)}_g[k]+2^{\lambda+1}-2$.
\item If $\lambda<h$: $2\Gamma^{(\lambda+1)}_g[k]\le\Gamma^{(\lambda)}_g[k]\le 2\Gamma^{(\lambda+1)}_g[k]+2$, and every
$i\in T_g(k)$ with $S^{(\lambda)}(k,i)=\Gamma^{(\lambda)}_g[k]$ has
$S^{(\lambda+1)}(k,i)\in\set{\Gamma^{(\lambda+1)}_g[k],\Gamma^{(\lambda+1)}_g[k]+1}$.
\end{enumerate}
\end{lemma}

\begin{proof}
(a) follows from $2^\lambda\floor{x/2^\lambda}\le x\le 2^\lambda\floor{x/2^\lambda}+2^\lambda-1$ for $x\ge 0$.
(b) follows by minimizing (a) over $T_g(k)$.
(c) By~\eqref{eq:a-bits}, $2S^{(\lambda+1)}\le S^{(\lambda)}\le 2S^{(\lambda+1)}+2$ pointwise; minimizing over
$T_g(k)$ gives the first claim. For a minimizer $i$ of $\Gamma^{(\lambda)}_g[k]$ we get
$2S^{(\lambda+1)}(k,i)\le\Gamma^{(\lambda)}_g[k]\le2\Gamma^{(\lambda+1)}_g[k]+2$, and
$S^{(\lambda+1)}(k,i)\ge\Gamma^{(\lambda+1)}_g[k]$ since $i\in T_g(k)$.
\end{proof}

\subsubsection{Segments}
\label{sec:a-segments}

Within a segment the high parts are constant, so the filled low parts are monotone and segments can be refined
by binary search.

\begin{definition}
Let $k\in\mathcal K^*$, let $(s,t)$ be a pair of parts, and let $\lambda\in[0,h]$. A \emph{level-$\lambda$ segment
of $(k,s,t)$} is a maximal subinterval of $I_k$ on which the four quantities
$\widetilde A_s[i]$, $\widetilde B_t[k-i]$, $A^{(\lambda)}_s[i]$, $B^{(\lambda)}_t[k-i]$ are constant.
For each $(k,s,t)$ these segments partition $I_k$. For a segment $\sigma$ we write
$\chi(\sigma)=\widetilde A_s[i]+\widetilde B_t[k-i]$ and $S^{(\lambda)}(\sigma)=A^{(\lambda)}_s[i]+B^{(\lambda)}_t[k-i]$
for any $i\in\sigma$. The segment is \emph{false for $g$} if $\chi(\sigma)\neq H[k]+g$.
If $T_g(k)\ne\emptyset$, the \emph{false-near set} $F^{(\lambda)}_g(k)$ is the set of all triples $(s,t,\sigma)$
such that $\sigma$ is a level-$\lambda$ segment of $(k,s,t)$ that is false for $g$ and satisfies
$|S^{(\lambda)}(\sigma)-\Gamma^{(\lambda)}_g[k]|\le1$.
\end{definition}

The symmetric window of radius~$1$ ensures that every retained level-$\lambda$ segment has a retained parent
at level $\lambda+1$ (\cref{lem:a-contain}). To compute the level-$\lambda$ minimum, only the parent values
$\Gamma^{(\lambda+1)}_g[k]$ and $\Gamma^{(\lambda+1)}_g[k]+1$ are needed (\cref{lem:a-levels}).

If $i\in V_k^{s,t}$ lies in the segment $\sigma$ of $(k,s,t)$, then the filled quantities agree with the actual
ones, so $\chi(\sigma)=\widetilde A[i]+\widetilde B[k-i]$ and $S^{(\lambda)}(\sigma)=S^{(\lambda)}(k,i)$.
In particular, every $i\in V_k^{s,t}\setminus T_g(k)$ lies in a segment of $(k,s,t)$ that is false for $g$, and
no element of $T_g(k)\cap V_k^{s,t}$ does. Segments may contain non-valid indices and indices of other types; not splitting segments at such indices keeps
their number small.

\begin{lemma}[Splitting]
\label{lem:a-split}
Let $\lambda<h$ and let $\sigma$ be a level-$(\lambda+1)$ segment of $(k,s,t)$. Then $\sigma$ is the disjoint
union of at most three level-$\lambda$ segments of $(k,s,t)$, which can be found by two binary searches.
Every level-$\lambda$ segment of $(k,s,t)$ is contained in exactly one level-$(\lambda+1)$ segment of $(k,s,t)$.
\end{lemma}

\begin{proof}
On $\sigma$, $\widetilde A_s[i]$ is constant, so $A^*_s[i]=\bar A_s[i]-p\widetilde A_s[i]$ is monotone in $i$, and
hence so is $A^{(\lambda)}_s[i]$, which takes values in $\set{2c,2c+1}$, where $c$ is the constant value of
$A^{(\lambda+1)}_s=\floor{A^{(\lambda)}_s/2}$ on $\sigma$. Thus the bit $A^{(\lambda)}_s[i]-2c$ is a monotone
$0/1$ function of $i\in\sigma$, and it changes at most once. Similarly, $B^{(\lambda)}_t[k-i]-2c'$ changes at
most once on $\sigma$. (The directions of these changes are irrelevant.) Hence $\sigma$ is the union of at most
three consecutive maximal subintervals on which all four level-$\lambda$ quantities are constant; they are
maximal since crossing an endpoint of $\sigma$ changes a level-$(\lambda+1)$ quantity, which is a function of the
level-$\lambda$ quantities. The last claim holds for the same reason.
\end{proof}

\begin{lemma}[Containment]
\label{lem:a-contain}
Let $\lambda<h$ and $T_g(k)\neq\emptyset$. If $(s,t,\tau)\in F^{(\lambda)}_g(k)$ and $\sigma$ is the
level-$(\lambda+1)$ segment of $(k,s,t)$ containing $\tau$, then $(s,t,\sigma)\in F^{(\lambda+1)}_g(k)$.
\end{lemma}

\begin{proof}
$\chi(\sigma)=\chi(\tau)$, so $\sigma$ is false for $g$. Let $v=S^{(\lambda+1)}(\sigma)$ and
$\Gamma=\Gamma^{(\lambda+1)}_g[k]$. By the filled analogue of~\eqref{eq:a-bits} and \cref{lem:a-levels}(c),
$2v\le S^{(\lambda)}(\tau)\le\Gamma^{(\lambda)}_g[k]+1\le2\Gamma+3$ and
$2v+2\ge S^{(\lambda)}(\tau)\ge\Gamma^{(\lambda)}_g[k]-1\ge2\Gamma-1$; since $v$ is an integer, $|v-\Gamma|\le1$.
\end{proof}

\subsubsection{The Algorithm}
\label{sec:a-algo}

\emph{The aliasing moduli.}
We choose $q_\lambda$ to balance polynomial multiplication against scanning: up to logarithmic factors, their
costs at level $\lambda$ will be $n2^{h-\lambda}q_\lambda$ and $Rn^2 2^\lambda/(\mathsf{P}q_\lambda)$, and since
$2^h=\Theta(\mathsf{P})$, this suggests $q_\lambda\approx\sqrt{nR}\,2^\lambda/\mathsf{P}$.
Formally, for $\lambda\in[0,h-1]$ let $\theta_\lambda=\sqrt{nR}\,2^\lambda/\mathsf{P}$. If $\theta_\lambda\ge32$, let
$\mathsf{Q}_\lambda=\ceil{\theta_\lambda}$ and let $q_\lambda$ be a uniformly random prime in $[\mathsf{Q}_\lambda,2\mathsf{Q}_\lambda]$;
otherwise let $q_\lambda=1$, so that every coarse sum passes the filter. All random choices are independent.

The algorithm maintains, for $\lambda=h,h-1,\dots,0$, a set $\mathrm{Act}_\lambda\subseteq\mathcal K^*\times\set{0,1}$
of \emph{active} pairs, numbers $\gamma^{(\lambda)}_g[k]$, and lists $L^{(\lambda)}_g(k)$ of triples
$(s,t,\sigma)$, where $\sigma$ is a level-$\lambda$ segment of $(k,s,t)$, stored with its endpoints,
$\chi(\sigma)$ and $S^{(\lambda)}(\sigma)$.
An active pair $(k,g)$ represents a coarse sum $H[k]+g$ still under consideration for output $k$.
By the invariant below, the stored value $\gamma^{(\lambda)}_g[k]$ equals the rounded minimum
$\Gamma^{(\lambda)}_g[k]$ whenever $T_g(k)\neq\emptyset$. At level $h$, a candidate with
$T_g(k)=\emptyset$ is initialized with $0$ and all its false segments; the first refinement removes it.
For other candidates, $L^{(\lambda)}_g(k)$ stores the false segments near the rounded minimum.

\begin{enumerate}[label=Step \arabic*.,leftmargin=*]
\item Compute the arrays $\pi_A,\pi_B$, the high and low parts, and $H$ (and hence $\mathcal K^*$).
The sequence $H$ is the $(\min,+)$-convolution of $(\widetilde A[i])_{i}$ and $(\widetilde B[j])_j$, with $+\infty$ at
non-finite positions; their finite entries are at most $U_*/p$, so we can use \cref{lem:fft}.
For every part, store the sorted list of positions at which $\widetilde A_s$ (resp.\ $\widetilde B_t$) changes;
there are at most $U_*/p+1$ of them per part, and the filled values can be evaluated by binary search in the
sorted list of the part.
\item (Level $h$.) Let $\mathrm{Act}_h=\mathcal K^*\times\set{0,1}$ and $\gamma^{(h)}_g[k]=0$, and let
$L^{(h)}_g(k)$ be the list of all triples $(s,t,\sigma)$ such that $\sigma$ is a level-$h$ segment of $(k,s,t)$
that is false for $g$.
\item For $\lambda=h-1,\dots,0$, with $q=q_\lambda$:
\begin{enumerate}[label=(\alph*)]
\item Compute, for all $k\in\mathcal K$, $d\in[0,2^{h-\lambda}]$, $c\in\set{0,1,2}$ and $\rho\in[0,2q-2]$, the number
\begin{multline*}
N(k,d,c,\rho)=\Bigl|\Bigl\{i\in V_k:\ S^{(\lambda+1)}(k,i)=d,\ a_\lambda[i]+b_\lambda[k-i]=c,\\
(\widetilde A[i]\bmod q)+(\widetilde B[k-i]\bmod q)=\rho\Bigr\}\Bigr|
\end{multline*}
as the coefficient of $z^ky^dx^cw^\rho$ in the product
\[
\Bigl(\sum_{i\in S_A}z^iy^{A^{(\lambda+1)}[i]}x^{a_\lambda[i]}w^{\widetilde A[i]\bmod q}\Bigr)\cdot
\Bigl(\sum_{j\in S_B}z^jy^{B^{(\lambda+1)}[j]}x^{b_\lambda[j]}w^{\widetilde B[j]\bmod q}\Bigr).
\]
These polynomials contain all finite entries, irrespective of their parts.
We use the convention that $N(k,d,c,\rho)=0$ for all arguments outside the stated ranges.
\item For each $(k,g)\in\mathrm{Act}_{\lambda+1}$, perform the following operations.

Let $\Gamma=\gamma^{(\lambda+1)}_g[k]$, $\mathcal D=\set{\Gamma,\Gamma+1}$ and
$r_0=(H[k]+g)\bmod q$. For $d\in\mathcal D$ and $c\in\set{0,1,2}$, let $N_q(k,d,c)=N(k,d,c,r_0)+N(k,d,c,r_0+q)$, omitting
the second term if $q=1$. (The sum of two residues lies in $[0,2q-2]$, so $r_0$ and $r_0+q$ are the only values
congruent to $r_0$ modulo $q$ that can occur.)

Split every $(s,t,\sigma)\in L^{(\lambda+1)}_g(k)$ into its level-$\lambda$ pieces
$\tau$ using \cref{lem:a-split}, and record $\chi(\tau)=\chi(\sigma)$ and
$S^{(\lambda)}(\tau)=2S^{(\lambda+1)}(\sigma)+c_\tau$, where $c_\tau\in\set{0,1,2}$ is the sum of the two filled bits
on $\tau$.

Call $(s,t,\sigma)$ \emph{aliased} if $q$ divides
$\chi(\sigma)-H[k]-g$. \emph{Scan} every piece $\tau$ of every aliased triple with $S^{(\lambda+1)}(\sigma)\in\mathcal D$,
that is, compute $|\tau\cap V^{s,t}_k|$ by testing for each $i\in\tau$ whether $i\in P_s$ and $k-i\in Q_t$. For
$d\in\mathcal D$ and $c\in\set{0,1,2}$, let $N^{\neg}(k,d,c)$ be the sum of $|\tau\cap V_k^{s,t}|$ over all scanned pieces with
$S^{(\lambda+1)}(\sigma)=d$ and $c_\tau=c$.

The difference $N_q(k,d,c)-N^\neg(k,d,c)$ counts the true pairs with parent-level value $d$ and bit sum $c$ (see
\cref{lem:a-inv}), each of which has next-level value $2d+c$. If $N_q(k,d,c)=N^\neg(k,d,c)$ for all $d\in\mathcal D$ and $c\in\set{0,1,2}$, declare
$(k,g)$ inactive; otherwise let
\[
\gamma^{(\lambda)}_g[k]=\min\set{2d+c:\ d\in\mathcal D,\ c\in\set{0,1,2},\ N_q(k,d,c)>N^\neg(k,d,c)}.
\]
\item (Pruning.) For each $k$ for which some pair $(k,g)$ is still active, let
$W_g=p(H[k]+g)+2^\lambda\gamma^{(\lambda)}_g[k]$ for each such $g$, and let $W^*_k=\min_gW_g$;
declare $(k,g)$ inactive if
$W_g>W^*_k+2^{\lambda+1}-2$. The remaining pairs form $\mathrm{Act}_\lambda$.
By \cref{lem:a-levels}, $C_g[k]\in[W_g,W_g+2^{\lambda+1}-2]$, so $W^*_k+2^{\lambda+1}-2$ is an upper bound on the best
remaining candidate, and a candidate whose lower bound exceeds it cannot attain $C[k]$. Pruning also keeps every
surviving $C_g[k]$ close to $C[k]$, which is needed in the running-time analysis.
\item For $(k,g)\in\mathrm{Act}_\lambda$, let $L^{(\lambda)}_g(k)$ be the list of pieces $(s,t,\tau)$ produced in (b)
from $L^{(\lambda+1)}_g(k)$ with $|S^{(\lambda)}(\tau)-\gamma^{(\lambda)}_g[k]|\le1$. This uses all pieces produced in~(b), including those that were not scanned: the aliasing filter only determines
which counts are needed at the current level, not which segments are retained.
\end{enumerate}
\item Output $C[k]=\infty$ for $k\notin\mathcal K^*$, and
$C[k]=\min\set{p(H[k]+g)+\gamma^{(0)}_g[k]:(k,g)\in\mathrm{Act}_0}$ for $k\in\mathcal K^*$.
\end{enumerate}

\subsubsection{Correctness}
\label{sec:a-correct}

We prove that an optimal candidate stays active and that each surviving candidate has the correct rounded
minimum and list of false segments.

\begin{lemma}[Invariant]
\label{lem:a-inv}
At level $h$: for every $(k,g)\in\mathrm{Act}_h$, the list $L^{(h)}_g(k)$ consists of all triples $(s,t,\sigma)$
such that $\sigma$ is a level-$h$ segment of $(k,s,t)$ that is false for $g$.
For every such $(k,g)$ with $T_g(k)\neq\emptyset$, we also have $L^{(h)}_g(k)=F^{(h)}_g(k)$ (since $S^{(h)}\equiv0=\Gamma^{(h)}_g[k]$),
$\gamma^{(h)}_g[k]=\Gamma^{(h)}_g[k]$, and $C_g[k]\le C[k]+2^{h+2}-4$.
For every $\lambda\in[0,h-1]$, after the iteration for $\lambda$:
\begin{enumerate}[label=(I\arabic*)]
\item for every $k\in\mathcal K^*$ there is $g$ with $(k,g)\in\mathrm{Act}_\lambda$ and $C_g[k]=C[k]$;
\item every $(k,g)\in\mathrm{Act}_\lambda$ has $T_g(k)\ne\emptyset$, $\gamma^{(\lambda)}_g[k]=\Gamma^{(\lambda)}_g[k]$, and
$L^{(\lambda)}_g(k)=F^{(\lambda)}_g(k)$;
\item every $(k,g)\in\mathrm{Act}_\lambda$ has $C_g[k]\le C[k]+2^{\lambda+2}-4$.
\end{enumerate}
These statements hold for every choice of $p$ and of the moduli $q_\lambda$. In particular, $\mathrm{Act}_\lambda$,
$\gamma^{(\lambda)}$ and $L^{(\lambda)}$ depend only on $p$: changing $q_\lambda$ changes the candidate counts and
the subtracted false counts by the same amount, so the refinement and pruning decisions are unaffected.
\end{lemma}

\begin{proof}
\emph{Level $h$.} The statement about the lists holds by Step~2. If $T_g(k)\neq\emptyset$ then
$\Gamma^{(h)}_g[k]=0=\gamma^{(h)}_g[k]$, and $C_g[k]\le p(H[k]+1)+2p-2$ while $C[k]\ge pH[k]$, so
$C_g[k]-C[k]\le3p-2\le2^{h+2}-4$.

\emph{Induction step.} Let $\lambda<h$, and assume the statements for level $\lambda+1$; for $\lambda+1=h$, (I1)
holds by \cref{lem:a-groups} since all pairs are active. Fix $(k,g)\in\mathrm{Act}_{\lambda+1}$ and let
$q=q_\lambda$, $d\in\mathcal D$ and $c\in\set{0,1,2}$. Since $(\widetilde A[i]\bmod q)+(\widetilde B[k-i]\bmod q)\in[0,2q-2]$
is congruent to $\widetilde A[i]+\widetilde B[k-i]$, and every element of $[0,2q-2]$ congruent to $r_0$ is either
$r_0$ or $r_0+q$ (an out-of-range second term contributes zero), the number $N_q(k,d,c)$ counts the $i\in V_k$ with $S^{(\lambda+1)}(k,i)=d$,
$a_\lambda[i]+b_\lambda[k-i]=c$ and $\widetilde A[i]+\widetilde B[k-i]\equiv H[k]+g\pmod q$. These $i$ are of two kinds:
those in $T_g(k)$, and the \emph{aliased false} ones, for which $\widetilde A[i]+\widetilde B[k-i]\neq H[k]+g$.

\emph{Case $T_g(k)\neq\emptyset$.} By induction, $\Gamma:=\gamma^{(\lambda+1)}_g[k]=\Gamma^{(\lambda+1)}_g[k]$ and
$L^{(\lambda+1)}_g(k)=F^{(\lambda+1)}_g(k)$. Let $d\in\mathcal D$, and let $i$ be aliased false with parameters $(d,c)$, of type
$(s,t)$. Then $i$ lies in a unique level-$(\lambda+1)$ segment $\sigma$ of $(k,s,t)$, which is false for $g$, aliased,
and has $S^{(\lambda+1)}(\sigma)=d$ with $|d-\Gamma|\le1$. Hence $(s,t,\sigma)\in F^{(\lambda+1)}_g(k)$ is split and
its pieces are scanned, and $i$ lies in exactly one piece $\tau$, for which $c_\tau=c$. Conversely, every
$i\in\tau\cap V_k^{s,t}$ for a scanned piece $\tau$ of $(s,t,\sigma)$ is valid, aliased false, and has the recorded
parameters $(d,c_\tau)$. Since every valid $i$ has exactly one type, $N^\neg(k,d,c)$ is exactly the number of
aliased false $i$ with parameters $(d,c)$, and therefore
\[
N_q(k,d,c)-N^\neg(k,d,c)=\bigl|\set{i\in T_g(k):\ S^{(\lambda+1)}(k,i)=d,\ a_\lambda[i]+b_\lambda[k-i]=c}\bigr|.
\]
By~\eqref{eq:a-bits} these $i$ have $S^{(\lambda)}(k,i)=2d+c$. By \cref{lem:a-levels}(c), some minimizer of
$\Gamma^{(\lambda)}_g[k]$ has $S^{(\lambda+1)}\in\mathcal D$, and every $i\in T_g(k)$ has $S^{(\lambda)}(k,i)\ge\Gamma^{(\lambda)}_g[k]$.
Hence $(k,g)$ is not declared inactive in~(b), and $\gamma^{(\lambda)}_g[k]=\Gamma^{(\lambda)}_g[k]$.

\emph{Case $T_g(k)=\emptyset$.} By (I2) this is only possible if $\lambda+1=h$. Then every $i$ counted by
$N_q(k,d,c)$ is aliased false, and lies in a level-$h$ segment of its type that is false for $g$ and aliased, with
$S^{(h)}=0$. All such triples are in $L^{(h)}_g(k)$, and $\Gamma=0$, so the same argument gives
$N_q(k,0,c)=N^\neg(k,0,c)$, while $N_q(k,1,c)=0=N^\neg(k,1,c)$ since $S^{(h)}\equiv0$. Hence $(k,g)$ is declared
inactive in~(b).

\emph{Pruning.} Every pair surviving~(b) satisfies $W_g\le C_g[k]\le W_g+2^{\lambda+1}-2$ by
\cref{lem:a-levels}(b). By (I1) at level $\lambda+1$, some active $g^*$ has $C_{g^*}[k]=C[k]$.
Since $T_{g^*}(k)\neq\emptyset$, it survives~(b). Moreover, $W_{g^*}\le C[k]=\min_gC_g[k]\le W^*_k+2^{\lambda+1}-2$,
so it also survives pruning. This shows (I1). Every survivor has $C_g[k]\le W_g+2^{\lambda+1}-2\le W^*_k+2^{\lambda+2}-4\le
C[k]+2^{\lambda+2}-4$, which shows (I3).

\emph{Lists.} For $(k,g)\in\mathrm{Act}_\lambda$, every piece produced from $F^{(\lambda+1)}_g(k)$ is a level-$\lambda$
segment that is false for $g$ (\cref{lem:a-split}). Since $\gamma^{(\lambda)}_g[k]=\Gamma^{(\lambda)}_g[k]$, every
triple retained by Step~3(d) belongs to $F^{(\lambda)}_g(k)$. Conversely, by \cref{lem:a-contain}, every triple of
$F^{(\lambda)}_g(k)$ arises as such a piece and passes the test in Step~3(d). Thus $L^{(\lambda)}_g(k)=F^{(\lambda)}_g(k)$,
which shows (I2).

Finally, all objects in (I1)--(I3) are defined in terms of $p$ only, and the pairs declared inactive in~(b) are
exactly those with $T_g(k)=\emptyset$, whatever $q_\lambda$ is.
\end{proof}

\begin{corollary}
\label{cor:a-correct}
The algorithm outputs $A\mc B$ for every choice of the random prime and auxiliary moduli.
\end{corollary}

\begin{proof}
For $k\notin\mathcal K^*$ the output is correct. For $k\in\mathcal K^*$, by (I2) at $\lambda=0$ every active $(k,g)$
has $p(H[k]+g)+\gamma^{(0)}_g[k]=C_g[k]\ge C[k]$, and by (I1) one of them attains $C[k]$.
\end{proof}

\subsubsection{Running Time}
\label{sec:a-time}

We bound the polynomial products, the number of tracked segments, and the total scanned length.
The moduli $q_\lambda$ control the product sizes. To bound the segment count, we compare segments with
\emph{cells} fixed before $p$ is sampled. The tracked lists depend only on $p$ (\cref{lem:a-inv}), so after
conditioning on $p$ we use the fresh modulus $q_\lambda$ to bound each segment's probability of being scanned.

Let $\mathcal L_\lambda=\sum_{(k,g)\in\mathrm{Act}_\lambda}|L^{(\lambda)}_g(k)|$ be the number of tracked triples at
level $\lambda$, and let $\mathrm{Scan}_\lambda$ be the total length of the pieces scanned in Step~3(b) at level
$\lambda$. Since $h$ depends on $p$, we sum over levels up to the deterministic bound
$h_{\max}=\ceil{\log_2(2\mathsf{P})}\ge h$, and set
$\mathcal L_\lambda=0$ for $h<\lambda\le h_{\max}$ and $\mathrm{Scan}_\lambda=0$ for $h\le\lambda\le h_{\max}$.
For each fixed $\lambda\in[0,h_{\max}]$, we bound $\mathbb E[\mathcal L_\lambda]$ and $\mathbb E[\mathrm{Scan}_\lambda]$,
with the expectation over all random choices, and then sum over these $O(\log n)$ levels.

\begin{lemma}[Deterministic costs]
\label{lem:a-detcost}
For every choice of the prime and auxiliary moduli, the running time is
\[
\Ot\Bigl(n\sqrt{RU_*}+\sum_{\lambda=0}^{h-1}n2^{h-\lambda}q_\lambda+\sum_{\lambda=0}^{h}\mathcal L_\lambda+\sum_{\lambda=0}^{h-1}\mathrm{Scan}_\lambda\Bigr).
\]
Moreover, for every $(k,s,t)$ and $\lambda$, the number of level-$\lambda$ segments of $(k,s,t)$ is at most
$3+2U_*/p+2(U_*+p)/2^\lambda$.
\end{lemma}

\begin{proof}
\emph{Number of segments.} Along $I_k$, the value $\widetilde A_s[i]$ changes at most $U_*/p$ times, and so does
$\widetilde B_t[k-i]$. On each of the at most $1+U_*/p$ maximal intervals on which $\widetilde A_s$ is constant,
$A^{(\lambda)}_s$ is monotone with values in $[0,(p-1)/2^\lambda]$, so in total it changes at most $(U_*+p)/2^\lambda$
times; similarly for $B^{(\lambda)}_t[k-\cdot]$. Every segment starts at $\min I_k$ or at a change.

\emph{Costs.} Step~1 takes $\Ot(n+\mathsf{P})$ time, plus $\Ot(nU_*/p)=\Ot(n\sqrt{U_*/R})$ for computing $H$ by \cref{lem:fft}.
Step~2 takes $\Ot(|\mathcal K|\cdot R\cdot(1+U_*/p))=\Ot(n\sqrt{RU_*})$ time, by merging the precomputed lists of change
positions of $\widetilde A_s$ and $\widetilde B_t$ for every $(k,s,t)$. In Step~3(a), the polynomials have $z$-degree
less than $n$, $y$-degree less than $2^{h-\lambda-1}$ (as $A^{(\lambda+1)}[i]\le(p-1)/2^{\lambda+1}$), $x$-degree at most
$1$ and $w$-degree less than $q$. Thus Kronecker substitution and the FFT (as in \cref{lem:fft}) compute their
product in $\Ot(n2^{h-\lambda}q_\lambda)$ time. Choosing $q_\lambda$ takes $\Ot(q_\lambda)$ time. Steps~3(b)--(d) take $O(\log^2 n)$ time per split triple and per piece: each binary search evaluates filled
values by another binary search in the sorted list of the part. Including $O(1)$ time per active pair and the
scanning cost, these steps take $\Ot(\mathcal L_{\lambda+1}+|\mathcal K|+\mathrm{Scan}_\lambda)$.
Step~4 takes $O(|\mathcal K|)$ time. Finally, $h=O(\log n)$.
\end{proof}

\begin{lemma}[Polynomial cost]
\label{lem:a-polycost}
$\sum_{\lambda<h}n2^{h-\lambda}q_\lambda=O(n\sqrt{RU_*}+n^{1.5}\sqrt R\log n)$.
\end{lemma}

\begin{proof}
The levels with $q_\lambda=1$ contribute at most $2n2^h\le 8n\mathsf{P}=O(n\sqrt{RU_*})$.
If $q_\lambda>1$, then $\theta_\lambda\ge32$ and $q_\lambda\le2\mathsf{Q}_\lambda\le3\theta_\lambda$, so
$n2^{h-\lambda}q_\lambda\le3n2^h\sqrt{nR}/\mathsf{P}\le12n^{1.5}\sqrt R$; there are $h=O(\log n)$ levels.
\end{proof}

\begin{lemma}[False segments give approximate multiples of $p$]
\label{lem:a-pos}
Let $\mu\in[0,h]$, $(k,g)\in\mathrm{Act}_\mu$ with $T_g(k)\neq\emptyset$, $(s,t,\sigma)\in F^{(\mu)}_g(k)$ and
$i\in\sigma$. Then $\Delta_{s,t}(k,i)=p\xi+\eta$, where $\xi=\chi(\sigma)-H[k]-g\neq0$ and $\eta$ is an integer
with $-3\cdot2^\mu<\eta<7\cdot2^\mu$.
\end{lemma}

\begin{proof}
Let $\xi=\chi(\sigma)-H[k]-g\neq0$, $\hat\Gamma=\Gamma^{(\mu)}_g[k]$ and $v=S^{(\mu)}(\sigma)$, so $|v-\hat\Gamma|\le1$.
By \cref{lem:a-levels}(a),(b), $\ell:=A^*_s[i]+B^*_t[k-i]\in[2^\mu v,2^\mu v+2^{\mu+1}-2]$ and
$\Gamma_g[k]\in[2^\mu\hat\Gamma,2^\mu\hat\Gamma+2^{\mu+1}-2]$, hence $|\ell-\Gamma_g[k]|<3\cdot2^\mu$.
Now $\mathrm{val}_{s,t}(k,i)-C_g[k]=p\xi+(\ell-\Gamma_g[k])$, and $0\le C_g[k]-C[k]\le2^{\mu+2}-4$ by
\cref{lem:a-groups} and \cref{lem:a-inv}. So $\eta:=\Delta_{s,t}(k,i)-p\xi=(\ell-\Gamma_g[k])+(C_g[k]-C[k])$ satisfies
$-3\cdot2^\mu<\eta<7\cdot2^\mu$.
\end{proof}

The next lemma is used at cell endpoints in \cref{lem:a-track} and at individual positions in \cref{lem:a-scan}.

\begin{lemma}[Divisibility near a fixed integer]
\label{lem:a-peps}
Let $\omega\ge1$ and $\mu\ge0$ be integers with $\omega\,2^\mu\le U_*$, and let $x_0\in[-2U_*,2U_*]$ be an integer that does not depend on $p$. The
probability that there is an integer $\eta$ with $|\eta|<\omega\,2^\mu$, $x_0\neq\eta$ and $p\mid x_0-\eta$ is
$O(\omega\,2^\mu\log n/\mathsf{P})$.
\end{lemma}

\begin{proof}
Each of the fewer than $2\omega\,2^\mu$ integers $x=x_0-\eta\ne0$ satisfies $|x|\le3U_*<\mathsf{P}^3$, so it is divisible by at most
two primes from $[\mathsf{P},2\mathsf{P}]$. By the prime-counting bound of \cref{lem:prime}, the probability that $p$ is one of them is
$O(\log \mathsf{P}/\mathsf{P})$. A union bound finishes the proof.
\end{proof}

Since $2^h<4\mathsf{P}\le U_*/256$, we have $\omega\,2^\mu\le U_*$ for all $\mu\le h$ and $\omega\le 64$.

\begin{lemma}[Tracked triples]
\label{lem:a-track}
For every fixed $\lambda\in[0,h_{\max}]$, $\mathbb E[\mathcal L_\lambda]=O(nRU_*\log n/\mathsf{P})=O(n\sqrt{RU_*}\log n)$.
\end{lemma}

\begin{proof}
Fix $k\in\mathcal K^*$, $g$, and a pair of parts $(s,t)$, and suppose $(k,g)\in\mathrm{Act}_\lambda$. If
$T_g(k)=\emptyset$, then $\lambda=h$ by \cref{lem:a-inv}, and the number of triples of type $(s,t)$ in
$L^{(h)}_g(k)$ is at most $5+4U_*/p$ by \cref{lem:a-detcost} (as $2^h\ge p$). Otherwise, we charge segment starts to cells defined from the filled values. Unlike segment endpoints, the cell
boundaries do not depend on $p$, so a fixed-integer divisibility bound applies there. Define a \emph{$\lambda$-cell}
of $(k,s,t)$
as a maximal subinterval $Z$ of $I_k$ on which both $\floor{\bar A_s[i]/2^\lambda}$ and
$\floor{\bar B_t[k-i]/2^\lambda}$ are constant, and let $e(Z)=\min Z$; $(k,s,t)$ has at most $1+2U_*/2^\lambda$ cells.
Let $(s,t,\tau)\in F^{(\lambda)}_g(k)$ and $i=\min\tau$. Let $Z$ be the cell containing $i$, and let $e=e(Z)$. Since $e$ and $i$
lie in the same cell, $|\bar A_s[i]-\bar A_s[e]|<2^\lambda$ and $|\bar B_t[k-i]-\bar B_t[k-e]|<2^\lambda$, so
$|\Delta_{s,t}(k,e)-\Delta_{s,t}(k,i)|<2^{\lambda+1}$. By \cref{lem:a-pos} applied to $(s,t,\tau)$ and $i$, we have $\Delta_{s,t}(k,e)=p\xi+r$ with
$\xi=\chi(\tau)-H[k]-g\ne0$ and $|r|<9\cdot2^\lambda$. Hence the event
\[
\mathcal E(k,s,t,Z):\quad\exists r,\ |r|<9\cdot2^\lambda,\ \Delta_{s,t}(k,e(Z))\neq r,\ p\mid\Delta_{s,t}(k,e(Z))-r
\]
occurs. A segment starting in $Z$ starts at $\min Z$, at one of the $\beta_Z$ indices
$i'\in Z\setminus\set{\min Z}$ where $\widetilde A_s$ or $\widetilde B_t[k-\cdot]$ changes, or at an index where
$A^{(\lambda)}_s$ or $B^{(\lambda)}_t[k-\cdot]$ changes while the corresponding high part does not. On each of the at most
$1+\beta_Z$ runs of $Z$ on which $\widetilde A_s$ is constant, $A^*_s$ is monotone and varies by less than $2^\lambda$,
so $A^{(\lambda)}_s$ changes at most once there, and similarly for $B$. Hence at most $3+3\beta_Z$ segments start in
$Z$, and $\sum_Z\beta_Z\le2U_*/p$. The contribution of $Z$ is at most $(3+3\beta_Z)\mathbf 1[\mathcal E(k,s,t,Z)]\le 3\,\mathbf 1[\mathcal E(k,s,t,Z)]+3\beta_Z$.
Therefore, for every choice of the prime and auxiliary moduli, the number of triples of type $(s,t)$ in
$F^{(\lambda)}_g(k)$ is at most $3\sum_Z\mathbf 1[\mathcal E(k,s,t,Z)]+6U_*/\mathsf{P}$.
The cells and the integers $\Delta_{s,t}(k,e(Z))$ do not depend on $p$. If $2^\lambda\ge4\mathsf{P}$, then $\lambda>h$ for
every $p$ and there is nothing to show; otherwise, by \cref{lem:a-peps} with $\omega=9$,
$\Pr[\mathcal E(k,s,t,Z)]=O(2^\lambda\log n/\mathsf{P})$, and summing over the at most $1+2U_*/2^\lambda$ cells gives
$O(U_*\log n/\mathsf{P})$. Summing over $k\in\mathcal K$, $g\in\set{0,1}$ and the $R$ pairs $(s,t)$ proves the claim.
\end{proof}

\begin{lemma}[Scanned length]
\label{lem:a-scan}
For every fixed $\lambda\in[0,h_{\max}]$, $\mathbb E[\mathrm{Scan}_\lambda]=O(n^{1.5}\sqrt R\log^2 n)$.
\end{lemma}

\begin{proof}
Condition on $p$, assume $\lambda<h$, and let $\mu=\lambda+1$. By \cref{lem:a-inv}, $\mathrm{Act}_\mu$ and the lists
$L^{(\mu)}$ depend only on $p$, and $q_\lambda$ is independent of $p$. Let $\varpi_\lambda=1$ if $\theta_\lambda<32$, and
$\varpi_\lambda=c_0\ln(U_*+1)/\mathsf{Q}_\lambda$ otherwise, where $c_0$ is the constant of \cref{lem:prime}. A triple $(s,t,\sigma)\in L^{(\mu)}_g(k)$ is scanned only if $q_\lambda$ divides
$\xi=\chi(\sigma)-H[k]-g$, where $0<|\xi|\le2U_*/p+1\le U_*$; by \cref{lem:prime} this has conditional probability at most $\varpi_\lambda$. Hence
\[
\mathbb E[\mathrm{Scan}_\lambda\mid p]\le\varpi_\lambda\sum_{(k,g)\in\mathrm{Act}_\mu}\ \sum_{(s,t,\sigma)\in L^{(\mu)}_g(k)}|\sigma|.
\]
For pairs $(k,g)$ with $T_g(k)\neq\emptyset$, the segments of a fixed type $(s,t)$ in $F^{(\mu)}_g(k)$ are disjoint, and
by \cref{lem:a-pos} each of their positions $i$ satisfies the event
\[
\mathcal E_\mu(k,s,t,i):\quad\exists\eta,\ |\eta|<7\cdot2^\mu,\ \Delta_{s,t}(k,i)\ne\eta,\ p\mid\Delta_{s,t}(k,i)-\eta.
\]
So their contribution is at most $\varpi_\lambda\sum_g\sum_k\sum_{s,t}|\set{i\in I_k:\mathcal E_\mu(k,s,t,i)}|$.
Pairs with $T_g(k)=\emptyset$ occur only if $\mu=h$, and contribute at most $\varpi_\lambda\sum_k\sum_{s,t}|I_k|\le\varpi_\lambda Rn^2$.

Now fix $\lambda$ and take the expectation over $p$. The quantities $\theta_\lambda$ and $\varpi_\lambda$, as well as
$\mathsf{Q}_\lambda$ when $\theta_\lambda\ge32$, do not depend on $p$. If $2^{\mu}\ge4\mathsf{P}$, then $\lambda\ge h$ for every $p$, so assume $2^\mu<4\mathsf{P}$. Since the integers
$\Delta_{s,t}(k,i)$ do not depend on $p$, \cref{lem:a-peps} with $\omega=7$ gives
$\Pr[\mathcal E_\mu(k,s,t,i)]=O(\min\set{1,2^\lambda\log n/\mathsf{P}})$. Using $\sum_k|I_k|\le n^2$, we obtain
\[
\mathbb E[\mathrm{Scan}_\lambda]=O\Bigl(\varpi_\lambda Rn^2\min\Bigl\{1,\frac{2^\lambda\log n}{\mathsf{P}}\Bigr\}\Bigr)+\varpi_\lambda Rn^2\Pr[\lambda=h-1].
\]
The two terms account for the pairs with $T_g(k)\ne\emptyset$ and $T_g(k)=\emptyset$, respectively.
If $\theta_\lambda\ge32$, then $\varpi_\lambda\le c_0\ln(U_*+1)\cdot \mathsf{P}/(\sqrt{nR}\,2^\lambda)$, and the first term is
$O(n^{1.5}\sqrt R\log^2n)$. For the second term, if $\Pr[\lambda=h-1]>0$, then some admissible $p$ has
$2^\lambda=2^{h-1}\ge p/2\ge \mathsf{P}/2$, so $\theta_\lambda\ge\sqrt{nR}/2$ and $\varpi_\lambda\le2c_0\ln(U_*+1)/\sqrt{nR}$, which bounds the
second term by $2c_0n^{1.5}\sqrt R\ln(U_*+1)$.

If $\theta_\lambda<32$, then $\varpi_\lambda=1$ and $2^\lambda<32\mathsf{P}/\sqrt{nR}$, so the first term is
$O(Rn^2\log n/\sqrt{nR})=O(n^{1.5}\sqrt R\log n)$. For the second term, if $\Pr[\lambda=h-1]>0$, then $2^\lambda\ge \mathsf{P}/2$
as before, hence $\sqrt{nR}<64$ and $Rn^2=n^{1.5}\sqrt R\cdot\sqrt{nR}<64n^{1.5}\sqrt R$.
\end{proof}

\begin{proof}[Proof of \cref{thm:nearly-mono-1}]
Correctness follows from \cref{cor:a-correct}. By \cref{lem:a-detcost,lem:a-polycost,lem:a-track,lem:a-scan},
the expected running time, summed over the $O(\log n)$ levels, is
$\Ot(n\sqrt{RU_*}+n^{1.5}\sqrt R)=\Ot(n^{1.5}\sqrt{m_Am_B})$, since $U_*=O(n)$ and $R=O(m_Am_B\log^2n)$. Together
with the $O(n)$ time for empty supports and the time of \cref{lem:greedy-partition} for computing the partitions
if they are not given, this gives $\Ot(n+n^{1.5}\sqrt{m_Am_B})$.
\end{proof}

\begin{remark}
\label{rem:a-large-values}
All estimates in the proof of \cref{thm:nearly-mono-1} remain valid for $U\le n^{O(1)}$: with $U_*=\max\set{U,n}$ we
still have $\log U_*=O(\log n)$, while the dependence on $U_*$ is retained explicitly until the final simplification.
Hence, for finite entries in $\set{0,\dots,U}$, the same algorithm and analysis give expected time
$\Ot(n+n\sqrt{m_Am_BU_*}+n^{1.5}\sqrt{m_Am_B})$, where $U_*=\max\set{U,n}$. For $m_A=m_B=1$, this is
$\Ot(n^{1.5}+n\sqrt U)$, matching the rectangular bound of Bringmann, D\"urr, and Polak~\cite{BringmannDP24} for
$U\ge n$ while allowing arbitrarily placed infinite entries.
\end{remark}

\begin{remark}
\label{rem:a-onesided}
The term $n^{1.5}\sqrt R$ comes from the trade-off between the polynomial size $n2^{h-\lambda}q_\lambda$ and the
expected scanning cost $\approx Rn^2 2^\lambda/(\mathsf{P}q_\lambda)$, whose product is $\Theta(Rn^3)$ independently of
$q_\lambda$. If the parts of one of the two sequences, say $B$, are \emph{intervals} (as is the case for a single
monotone sequence without infinities), then $|\tau\cap V_k^{s,t}|$ is a difference of two prefix sums of the
indicator of $P_s$, so no scanning and no aliasing filter are needed, and the running time becomes
$\Ot(n\sqrt{RU_*})$. Since $U_*\ge n$, this matches the bound of \cref{rem:a-large-values}, so the interval case gives a simpler
algorithm but no asymptotic improvement.
\end{remark}

\section{From Bounded to Monotone \texorpdfstring{$(\min,+)$}{(min,+)}-Convolution}
\label{sec:mono-red}

The reduction in \cref{thm:intro-mono}, illustrated in \cref{fig:mono-red}, relies on the following identity for
floors.

\begin{lemma}
\label{lem:floor}
Let $w\ge 1$ be an integer, and let $x,y,s$ be integers. Let $J$ be a set of at least $w$ consecutive integers. Then
\[
\min_{j\in J}\left(\floor{\frac{x+j}{w}}+\floor{\frac{y+s-j}{w}}\right)=\floor{\frac{x+y+s+1}{w}}-1.
\]
\end{lemma}

\begin{proof}
Let $T=x+y+s$, and consider any integers $a,b$ with $a+b=T$.
Write $a=q_1w+r_1$ and $b=q_2w+r_2$ with $0\le r_1,r_2<w$.
Then $T+1=(q_1+q_2)w+r_1+r_2+1$ with $1\le r_1+r_2+1\le 2w-1$, so
$\floor{(T+1)/w}\le q_1+q_2+1$, that is,
$\floor{a/w}+\floor{b/w}\ge\floor{(T+1)/w}-1$.
This proves that the left-hand side is at least the right-hand side.
Conversely, since $J$ contains $w$ consecutive integers, there is $j\in J$ with $x+j\equiv w-1\pmod w$.
For $a=x+j$ and $b=y+s-j$ we then have $r_1=w-1$, hence $\floor{a/w}=(a-w+1)/w$ and
\[
\floor{\frac{a}{w}}+\floor{\frac{b}{w}}=\floor{\frac{a-w+1+b}{w}}=\floor{\frac{T+1}{w}}-1. \qedhere
\]
\end{proof}

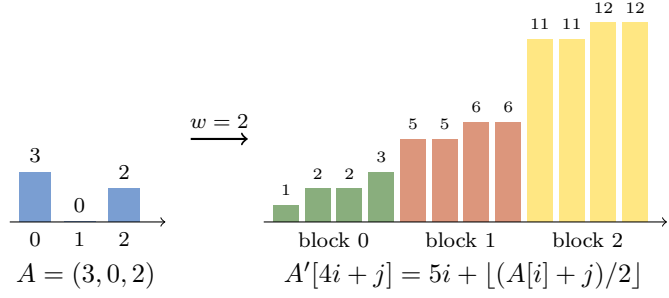
\begin{figure}[t]
\centering
\begin{tikzpicture}[x=0.42cm,y=0.22cm]
  \foreach \i/\v in {0/3,1/0,2/2} {
    \fill[NavyBlue!55] (\i*1.4,0) rectangle ++(1,\v);
    \node[below,font=\scriptsize] at (\i*1.4+0.5,0) {$\i$};
    \node[above,font=\scriptsize] at (\i*1.4+0.5,\v) {$\v$};
  }
  \draw[->] (-0.3,0) -- (4.6,0);
  \node[font=\small] at (2.1,-3.2) {$A=(3,0,2)$};
  \draw[->,thick] (5.4,5) -- (7.2,5) node[midway,above,font=\scriptsize] {$w=2$};
  \foreach \t/\v in {0/1,1/2,2/2,3/3,4/5,5/5,6/6,7/6,8/11,9/11,10/12,11/12} {
    \pgfmathtruncatemacro{\blk}{int(\t/4)}
    \ifnum\blk=0 \def\col{OliveGreen!55}\fi
    \ifnum\blk=1 \def\col{BrickRed!45}\fi
    \ifnum\blk=2 \def\col{Goldenrod!70}\fi
    \fill[\col] (8+\t,0) rectangle ++(0.8,\v);
    \node[above,font=\tiny] at (8.4+\t,\v) {$\v$};
  }
  \draw[->] (7.7,0) -- (20.4,0);
  \foreach \b in {0,1,2} { \node[below,font=\scriptsize] at (9.9+4*\b,0) {block $\b$}; }
  \node[font=\small] at (14,-3.2) {$A'[4i+j]=5i+\lfloor (A[i]+j)/2\rfloor$};
\end{tikzpicture}
\caption{The reduction of \cref{thm:mono-red} for $A=(3,0,2)$, $U=3$ and $w=2$, so that $K=2\lfloor U/w\rfloor+3=5$.
Every entry $A[i]$ becomes a non-decreasing block of $2w$ entries $iK+\lfloor(A[i]+j)/w\rfloor$, and consecutive
blocks are separated by the offset $K$, so the resulting sequence is monotone.}
\label{fig:mono-red}
\end{figure}

\begin{theorem}
\label{thm:mono-red}
Let $A,B\in\set{0,\dots,U}^n$ and let $w\ge 1$ be an integer.
Set $K=2\floor{U/w}+3$ and define the sequences $A',B'$ of length $N=2nw$ by
\[
A'[2wi+j]=iK+\floor{\frac{A[i]+j}{w}},\qquad B'[2wi+j]=iK+\floor{\frac{B[i]+j}{w}},
\qquad i\in[n],\ j\in[2w].
\]
Then $A'$ and $B'$ are non-decreasing with entries in $\set{0,\dots,nK}$, and for every $m\in[2n-1]$ and
$u\in\set{1,\dots,w}$,
\begin{equation}
\label{eq:readout}
(A'\mc B')[2w(m+1)+u-1]=mK+1+\floor{\frac{(A\mc B)[m]+u}{w}}.
\end{equation}
Consequently, $A\mc B$ can be computed from $A'\mc B'$ in $O(nw)$ time.
\end{theorem}

\begin{proof}
Write $C=A\mc B$ and $C'=A'\mc B'$.
Let $q=\floor{U/w}$, so that $K=2q+3$.
For every $x\in\set{0,\dots,U}$ and $j\in[2w]$ we have
$0\le\floor{(x+j)/w}\le\floor{(U+2w-1)/w}\le q+2$.
Let $F=q+2$, so that $F<K$.

Within the block of index $i$, the entries $iK+\floor{(A[i]+j)/w}$ are non-decreasing in $j$.
The last entry of block $i$ is at most $iK+F<(i+1)K$, which is at most the first entry of block $i+1$.
Hence $A'$ is non-decreasing, and its entries lie in $\set{0,\dots,(n-1)K+F}\subseteq\set{0,\dots,nK}$.
The same holds for $B'$.

Fix $m\in[2n-1]$ and $u\in\set{1,\dots,w}$, and let $\kappa=2w(m+1)+u-1$.
A pair of positions $(2wi+j,\,2wi'+j')$ with $i,i'\in[n]$ and $j,j'\in[2w]$ contributes to $C'[\kappa]$ if and
only if $2w(i+i')+(j+j')=\kappa$.
Since $0\le j+j'\le 4w-2$ and $u-1\le w-1$, this happens exactly in the following two cases:
\begin{enumerate}[label=(\roman*)]
\item $i+i'=m+1$ and $j+j'=u-1$; such a pair has value at least $(m+1)K$;
\item $i+i'=m$ and $j+j'=2w+u-1$; such a pair has value
$mK+\floor{(A[i]+j)/w}+\floor{(B[i']+j')/w}$.
\end{enumerate}
Pairs of type~(ii) exist: since $m\in[2n-1]$ there are $i,i'\in[n]$ with $i+i'=m$, and then every
$j\in J:=\set{u,u+1,\dots,2w-1}$ gives $j'=2w+u-1-j\in[2w]$.
Let $\Phi$ denote the minimum value of a type~(ii) pair. We have
\[
\Phi=mK+\min_{\substack{i+i'=m}}\ \min_{j\in J}\left(\floor{\frac{A[i]+j}{w}}+\floor{\frac{B[i']+2w+u-1-j}{w}}\right).
\]
Since $|J|=2w-u\ge w$, applying \cref{lem:floor} with $x=A[i]$, $y=B[i']$, and $s=2w+u-1$ shows that the inner
minimum equals $\floor{(A[i]+B[i']+2w+u)/w}-1=\floor{(A[i]+B[i']+u)/w}+1$.
As $t\mapsto\floor{(t+u)/w}$ is non-decreasing, minimizing over $i+i'=m$ gives
$\Phi=mK+1+\floor{(C[m]+u)/w}$.
Finally, since $C[m]\le 2U$ and $u\le w$, we have $\Phi\le mK+1+\floor{(2U+w)/w}\le mK+1+(2q+2)=(m+1)K$, so
no type~(i) pair has a smaller value than $\Phi$. Hence $C'[\kappa]=\Phi$, which proves~\eqref{eq:readout}.

Fix $m$, and let $f(u)=\floor{(C[m]+u)/w}$ for $u\in\set{1,\dots,w}$; by~\eqref{eq:readout} all these
values can be read off from $C'$.
Write $C[m]=q_mw+r_m$ with $0\le r_m<w$. Then $f(u)=q_m+\mathbf 1[u\ge w-r_m]$; in particular $f(w)=q_m+1$, and
the smallest $u\in\set{1,\dots,w}$ with $f(u)=f(w)$ is $u^*=w-r_m$.
Hence $C[m]=w\cdot f(w)-u^*$, which can be computed in $O(w)$ time per $m$.
\end{proof}

\begin{remark}
\label{rem:dec}
Reversing the index order of both constructed sequences gives non-increasing instances and reverses the output
order. Negating both sequences and adding constant offsets converts the construction into one for
$(\max,+)$-convolution. Recovering $A\mc B$ uses only the $(2n-1)w$ entries of $A'\mc B'$ listed
in~\eqref{eq:readout}.
\end{remark}

\begin{corollary}
\label{cor:four-thirds}
Suppose that $(\min,+)$-convolution on non-decreasing sequences of length $N$ with entries in $\set{0,\dots,cN}$ can be
solved in time $T(N)$, for a sufficiently large constant $c$, where $T(N)\ge N$ is non-decreasing.
Then $(\min,+)$-convolution on sequences of length $n$ with entries in $\set{0,\dots,n}$ can be solved in time
$T(2n\ceil{\sqrt n})+O(n^{1.5})$.
In particular, if $T(N)=O(N^{4/3-\eps})$ for some $\eps>0$, then bounded $(\min,+)$-convolution can be solved in
$O(n^{2-3\eps/2})$ time.
\end{corollary}

\begin{proof}
Apply \cref{thm:mono-red} with $U=n$ and $w=\ceil{\sqrt n}$. The resulting sequences have length
$N=2n\ceil{\sqrt n}$ and entries at most $n(2\floor{n/w}+3)\le 2n^{1.5}+3n\le cN$.
The remaining claims follow from $N=\Theta(n^{1.5})$. Since $T(N)\ge N$ forces $\eps\le1/3$, the
$O(n^{1.5})$ construction cost is dominated by $O(n^{2-3\eps/2})$.
\end{proof}

Choosing $w=\ceil{\sqrt n}$ balances length against value range. More generally, $w=\ceil{n^\gamma}$ gives the
trade-off below: $\gamma=1/2$ recovers \cref{cor:four-thirds}, and $\gamma=0$ (so $w=1$) is essentially the
linear-offset reduction that adds $iK=i(2U+3)$ to $A[i]$.

\begin{corollary}
\label{cor:tradeoff}
Let $\gamma\in[0,1)$ be a constant.
If $(\min,+)$-convolution on non-decreasing sequences of length $N=\Theta(n^{1+\gamma})$ with entries in
$\set{0,\dots,O(n^{2-\gamma})}$ can be solved in $O(n^{2-\eps})$ time for some $\eps>0$, then
$(\min,+)$-convolution on sequences of length $n$ with entries in $\set{0,\dots,n}$ can be solved in
$O(n^{2-\eps}+n^{1+\gamma})$ time.
\end{corollary}

\begin{proof}
Apply \cref{thm:mono-red} with $U=n$ and $w=\ceil{n^{\gamma}}$. The length is $2nw=\Theta(n^{1+\gamma})$
and the entries are at most $n(2n/w+3)=O(n^{2-\gamma})$. The construction and readout take $O(n^{1+\gamma})$ time,
which is truly subquadratic since $\gamma<1$.
\end{proof}

In terms of the length $N=\Theta(n^{1+\gamma})$ and the value bound $V=\Theta(n^{2-\gamma})$, we have
$NV=\Theta(n^3)$. So by \cref{cor:tradeoff}, for every fixed $\gamma\in[0,1)$, that is, roughly in the range
$N^{1/2}<V\le N^{2}$, an $O((NV)^{2/3-\eps})$-time algorithm for monotone $(\min,+)$-convolution would give an
$O(n^{2-3\eps}+n^{1+\gamma})$-time algorithm for bounded inputs of length $n$.
For comparison, the algorithms of~\cite{ChiDXZ22,BringmannDP24} run in expected time $\Ot(N\sqrt V)$.

\section{Bounded Value Ranges}
\label{sec:values}

\subsection{Reducing the Value Range}
\label{sec:value-red}

We prove \cref{thm:intro-value} by splitting the entries into value layers and concatenating the layers.

\begin{theorem}
\label{thm:value-red}
Let $A,B\in\set{0,\dots,U}^n$, let $1\le L\le U$ be an integer, and let $q=\floor{U/L}+1$.
Define sequences $A',B'$ of length $N=2nq$ as follows: for $c\in[q]$ and $i\in[2n]$, let
\[
A'[2nc+i]=\begin{cases}
A[i]-cL & \text{if } i<n \text{ and } \floor{A[i]/L}=c,\\
2L & \text{otherwise},
\end{cases}
\]
and define $B'$ analogously.
Then $A',B'\in\set{0,\dots,2L}^N$, and for every $k\in[2n-1]$,
\begin{equation}
\label{eq:value-readout}
(A\mc B)[k]=\min\set{dL+(A'\mc B')[2nd+k]\;:\; d\in[2q-1],\ (A'\mc B')[2nd+k]<2L}.
\end{equation}
In particular, $A\mc B$ can be computed from $A'\mc B'$ in $O(N)$ time.
\end{theorem}

For example, for $A=(1,4,2,0)$, $U=4$, and $L=2$, we have $q=3$ layers of length $2n=8$, and
\[
A'=(\underbrace{1,4,4,0,4,4,4,4}_{c=0},\ \underbrace{4,4,0,4,4,4,4,4}_{c=1},\ \underbrace{4,0,4,4,4,4,4,4}_{c=2}),
\]
where every entry equal to $2L=4$ is a dummy.

\begin{proof}
Every entry of $A'$ is either $2L$ or of the form $A[i]-cL\in[0,L)$, so $A'\in\set{0,\dots,2L}^N$;
similarly for $B'$.
Call a position $2nc+i$ of $A'$ \emph{real} if $i<n$ and $\floor{A[i]/L}=c$, and a \emph{dummy} otherwise; every
$i\in[n]$ has exactly one real position, namely in layer $c=\floor{A[i]/L}$.
Layers have length $2n$ so that the within-layer offsets of two real positions sum to less than $2n$.

Fix $d\in[2q-1]$ and $k\in[2n-1]$, and let $\kappa=2nd+k$.
Consider a pair of positions $(2nc+i,\,2nc'+i')$ with $2n(c+c')+i+i'=\kappa$.
If both positions are real, then $i,i'<n$, so $i+i'\le 2n-2$ and therefore $c+c'=d$ and $i+i'=k$; the value
of the pair is $A[i]+B[i']-dL$, which is less than $2L$ because $A[i]-cL<L$ and $B[i']-c'L<L$.
If at least one position is a dummy, the value of the pair is at least $2L$, since all entries are
nonnegative.
Hence $(A'\mc B')[\kappa]<2L$ holds if and only if there is a pair $(i,i')$ with $i+i'=k$ and
$\floor{A[i]/L}+\floor{B[i']/L}=d$, and in this case
\[
dL+(A'\mc B')[\kappa]=\min\set{A[i]+B[i'] \;:\; i+i'=k,\ \floor{A[i]/L}+\floor{B[i']/L}=d}.
\]
Every pair $(i,i')$ with $i+i'=k$ has $\floor{A[i]/L}+\floor{B[i']/L}\in[2q-1]$, so taking the minimum over
$d$ proves~\eqref{eq:value-readout}.
\end{proof}

\begin{corollary}
\label{cor:value-lb}
Let $x\in(0,1]$ be a constant, and suppose that $(\min,+)$-convolution on sequences of length $N$ with entries in
$\set{0,\dots,N^{x}}$ can be solved in time $O(N^{1+x-\eps})$ for some $\eps>0$.
Then bounded $(\min,+)$-convolution on sequences of length $n$ can be solved in time $O(n^{2-2\eps/(1+x)})$.

\end{corollary}

\begin{proof}
Let $A,B\in\set{0,\dots,n}^n$. Apply \cref{thm:value-red} with $U=n$ and
$L=\max\bigl(1,\floor{\tfrac12 n^{2x/(1+x)}}\bigr)\le n$.
The resulting length is $N=2n(\floor{n/L}+1)=\Theta(n^{2/(1+x)})$; moreover, for $n$ large enough,
$N\ge 2n^2/L\ge 4n^{2/(1+x)}$, so the entries are at most $2L\le n^{2x/(1+x)}\le N^{x}$.
The running time is $O(N^{1+x-\eps})=O(n^{2-2\eps/(1+x)})$, plus $O(N)=O(n^{2/(1+x)})$ for the
reduction, which is truly subquadratic since $x>0$. Since writing the output takes linear time, we may assume
$\eps\le x$, in which case the first term dominates.
\end{proof}

\subsection{One Bounded Sequence Suffices}
\label{sec:one-bounded}

We now prove \cref{thm:intro-one-bounded} by splitting the instance into subinstances that are covered by
\cref{lem:fft}.

\begin{theorem}
\label{thm:one-bounded}
Let $A\in\set{0,\dots,M}^{n}$ and $B\in(\Zinf)^{n}$. Then $A\mc B$ can be computed deterministically in
$O(n(M+1)\log n)$ time.
\end{theorem}

Write $C=A\mc B$. It suffices to compute the outputs $C[k]$ for $k\in[n]$:
the outputs with $k\ge n-1$ follow by applying the same procedure to the reversed sequences
$A^R[i]=A[n-1-i]$ and $B^R[j]=B[n-1-j]$, since $(A^R\mc B^R)[k]=C[2n-2-k]$.
If $M=0$ then $C[k]=\min_{j\le k}B[j]$ for $k<n$, so assume $M\ge1$.

Let $B_{\min}[j]=\min_{j'\le j}B[j']$ denote the prefix minima of $B$.
If $B_{\min}[k]=\infty$, then $C[k]=\infty$. Let $j_0$ be the smallest index with $B_{\min}[j_0]<\infty$; if no
such index exists, then $C\equiv\infty$, and otherwise $C[k]=\infty$ for $k<j_0$.

\begin{lemma}
\label{lem:prediction}
For every $k$ with $j_0\le k<n$ we have $B_{\min}[k]\le C[k]\le B_{\min}[k]+M$. Moreover, every witness $(i,j)$ of $C[k]$
satisfies $B_{\min}[k]\le B_{\min}[j]\le B[j]\le B_{\min}[k]+M$.
\end{lemma}

\begin{proof}
Every pair $(i,j)$ with $i+j=k$ has $j\le k$, hence $A[i]+B[j]\ge B[j]\ge B_{\min}[j]\ge B_{\min}[k]$.
On the other hand, let $j^*\le k$ with $B[j^*]=B_{\min}[k]$ and $i^*=k-j^*$. Then $0\le i^*\le k<n$, so
$C[k]\le A[i^*]+B[j^*]\le M+B_{\min}[k]$.
Finally, a witness $(i,j)$ satisfies $B[j]\le A[i]+B[j]=C[k]\le B_{\min}[k]+M$, and $B_{\min}[k]\le B_{\min}[j]\le B[j]$
holds as $j\le k$.
\end{proof}

\begin{proof}[Proof of \cref{thm:one-bounded}]
For $j_0\le k<n$ define the \emph{level} $g(k)=\floor{(B_{\min}[j_0]-B_{\min}[k])/M}\ge 0$.
Since $B_{\min}$ is non-increasing, $g$ is non-decreasing, so for every $\gamma\ge0$ the set
$K_\gamma=\set{k:g(k)=\gamma}$ is a (possibly empty) interval.
Let $J_\gamma=K_{\gamma-1}\cup K_\gamma$ (with $K_{-1}=\emptyset$), which is an interval as well.
Fix $\gamma$ with $K_\gamma\neq\emptyset$ and $k\in K_\gamma$, and let $(i,j)$ be a witness of $C[k]$.
Since $B[j]<\infty$, we have $j\ge j_0$, so $g(j)$ is defined. By \cref{lem:prediction}, $B_{\min}[j]\le B_{\min}[k]+M$, so
$g(j)\ge \floor{(B_{\min}[j_0]-B_{\min}[k]-M)/M}=g(k)-1$. Since $j\le k$ and $g$ is
non-decreasing, also $g(j)\le g(k)$, so $j\in J_\gamma$.
Moreover, $B[j]\in[B_{\min}[k],B_{\min}[k]+M]$, and $B_{\min}[k]\in(B_{\min}[j_0]-(\gamma+1)M,\,B_{\min}[j_0]-\gamma M]$, so
$B[j]$ lies in the window $W_\gamma=(B_{\min}[j_0]-(\gamma+1)M,\ B_{\min}[j_0]-(\gamma-1)M]$ of length $2M$.

Let $s_\gamma=\min J_\gamma$ and $\ell_\gamma=|J_\gamma|$, and define the sequence
$B_\gamma$ of length $\ell_\gamma$ by $B_\gamma[t]=B[s_\gamma+t]$ if $B[s_\gamma+t]\in W_\gamma$, and
$B_\gamma[t]=\infty$ otherwise.
Let $A_\gamma=(A[0],\dots,A[\ell_\gamma-1])$.
We show that $C[k]=(A_\gamma\mc B_\gamma)[k-s_\gamma]$ for all $k\in K_\gamma$. Every pair $(t,t')$ contributing to $(A_\gamma\mc B_\gamma)[k-s_\gamma]$ with finite value
corresponds to the pair $(t,s_\gamma+t')$ for $C[k]$, so the right-hand side is at least $C[k]$.
Conversely, a witness $(i,j)$ of $C[k]$ has $j\in J_\gamma$ and $B[j]\in W_\gamma$ as shown above, and
$i=k-j\le\max J_\gamma-s_\gamma<\ell_\gamma$; so it corresponds to the pair $(i,j-s_\gamma)$ on the right-hand side.

The finite entries of $A_\gamma$ lie in $[0,M]$ and those of $B_\gamma$ in an interval of length $2M$, so by
\cref{lem:fft} the convolution $A_\gamma\mc B_\gamma$ can be computed in $O(\ell_\gamma (M+1)\log n)$ time.
An index in $K_\delta$ belongs only to $J_\delta$ and $J_{\delta+1}$, so $\sum_\gamma\ell_\gamma\le 2n$, and the total running
time is $O(n(M+1)\log n)$. Computing $B_{\min}$, the levels, and the intervals takes $O(n)$ time: we process only the nonempty levels, which
are the maximal runs of the non-decreasing sequence $g$.
\end{proof}

\begin{remark}
\label{rem:one-bounded-rect}
The proof uses that $A$ has no infinite entries: the comparison pair $(k-j^*,j^*)$ in the proof of
\cref{lem:prediction} must have finite value.
If $A$ has length $n_A$ and $B$ has length $n_B$, the same running time $\Ot((n_A+n_B)(M+1))$ is obtained by
padding $B$ with $+\infty$ entries when $n_B<n_A$, and by splitting $B$ into consecutive chunks of length $n_A$
(the last one padded with $+\infty$) when $n_B>n_A$; in the latter case, the output is the entrywise minimum of the
chunk outputs, each shifted by the starting index of its chunk.
\end{remark}

\subsection{Bounded Entries with Infinities}
\label{sec:bounded-inf}

To prove \cref{thm:intro-bounded-inf}, we replace the prefix-minimum prediction by a search for the first relevant
value bucket of $B$.

\begin{theorem}
\label{thm:bounded-inf}
Let $A\in(\set{0,\dots,M}\cup\set{+\infty})^{n}$ and $B\in(\Zinf)^{n}$, where $1\le M\le n$.
Then $A\mc B$ can be computed deterministically in $O(n^{1.5}\sqrt{M}\log n)$ time.
\end{theorem}

\begin{proof}[Proof of \cref{thm:bounded-inf}]
Write $C=A\mc B$ and let $\tau=\ceil{\sqrt{nM}}$.
For every $j$ with $B[j]<\infty$ let $\beta(j)=\floor{B[j]/M}$ be its \emph{bucket}, and for every integer $b$
let $T_b=\set{j: B[j]<\infty,\ \beta(j)=b}$.
A bucket is \emph{heavy} if $|T_b|\ge\tau$ and \emph{light} otherwise; there are at most $n/\tau$ heavy buckets.
Sort the indices $j$ with $B[j]<\infty$ by $B[j]$, and partition the sorted order into \emph{blocks}.
Each heavy bucket forms its own block. In each maximal run of light buckets, greedily add whole buckets to a
block until it contains at least $\tau$ elements, then start a new block.
Each block consisting of light buckets has fewer than $2\tau$ elements, and all but at most one block in each
such run have at least $\tau$ elements. There are at most $n/\tau$ heavy buckets, hence at most $n/\tau+1$ runs of
light buckets, and at most $n/\tau$ blocks with at least $\tau$ elements; so there are $O(n/\tau)$ blocks.
By construction, blocks are ordered by value: all elements of an earlier block have smaller buckets than all
elements of a later block.

For an output $k$, call $j$ \emph{feasible} if $B[j]<\infty$, $0\le k-j<n$ and $A[k-j]<\infty$. If no $j$ is
feasible, $C[k]=\infty$. Otherwise, let $b^*(k)=\min\set{\beta(j): j \text{ feasible for } k}$.

\emph{Step 1: computing $b^*(k)$.} We follow~\cite[Lemma~1]{ChanH20}. For every block $\mathcal B$, we multiply
the indicator polynomials of $\set{i: A[i]<\infty}$ and of $\mathcal B$ with one FFT, which gives the number
$E_{\mathcal B}[k]$ of feasible $j\in \mathcal B$ for every $k$. The exponents are the original indices; only the grouping into blocks uses the sorted order.
This takes $O(n\log n)$ time per block, and $O((n^2/\tau)\log n)$ time in total.
For every $k$, the first block $\mathcal B(k)$ with $E_{\mathcal B(k)}[k]>0$ contains a feasible $j$ of minimal bucket,
because blocks are ordered by value and consist of whole buckets.
If $\mathcal B(k)$ is a heavy bucket $b$, then $b^*(k)=b$. Otherwise, we scan the elements of $\mathcal B(k)$ in sorted order and
return the bucket of the first feasible one; this takes $O(\tau)$ time.
In total, Step~1 takes $O(n^2\log n/\tau+n\tau)$ time.

\emph{Step 2: the relevant buckets.}
Let $k$ have a feasible $j$, and let $b^*=b^*(k)$.
Every feasible $j$ has $B[j]\ge b^*M$, so $C[k]\ge b^*M$. A feasible $j$ with $\beta(j)=b^*$ has
$B[j]\le b^*M+M-1$ and $A[k-j]\le M$, so $C[k]\le (b^*+2)M-1$.
A feasible $j$ with $\beta(j)\ge b^*+2$ has $A[k-j]+B[j]\ge(b^*+2)M>C[k]$.
Hence every witness $(k-j,j)$ satisfies $\beta(j)\in\set{b^*,b^*+1}$, and
\[
C[k]=\min_{b\in\set{b^*,b^*+1}}\ \min\set{A[k-j]+B[j]\;:\;j\in T_b,\ 0\le k-j<n}.
\]

\emph{Step 3: evaluating a bucket.}
For every heavy bucket $b$, precompute $D_b=A\mc B_b$, where $B_b[j]=B[j]$ for $j\in T_b$ and $B_b[j]=\infty$
otherwise. The finite entries of $B_b$ lie in $[bM,bM+M-1]$, so by \cref{lem:fft} this takes
$O(nM\log n)$ time per heavy bucket, and $O((n^2M/\tau)\log n)$ time in total.
For every $k$ having a feasible index and every $b\in\set{b^*(k),b^*(k)+1}$, the inner minimum above is
$D_b[k]$ if $b$ is heavy, and we compute it by brute force over $T_b$ in $O(\tau)$ time if $b$ is light. With the
buckets stored in a sorted array, this takes $O(n\tau\log n)$ time in total.

The total running time is $O\bigl((n^2M/\tau+n\tau)\log n\bigr)=O(n^{1.5}\sqrt M\log n)$.
\end{proof}

\begin{remark}
\label{rem:min-witness-conv}
In the \emph{minimum-witness convolution} problem we are given sets $S,T\subseteq[n]$ and must compute, for
every $k$, the smallest $j\in T$ with $k-j\in S$. This is the special case of \cref{thm:bounded-inf} with
$M=1$, $A[i]=0$ for $i\in S$, $A[i]=\infty$ otherwise, and $B[j]=j$ for $j\in T$, $B[j]=\infty$ otherwise.
The best known algorithms for minimum-witness convolution take $\Ot(n^{1.5})$
time~\cite{BabaiF09,ChanH20} (see also~\cite{LincolnPW20}), so
an $O(n^{1.5-\eps})$-time algorithm for the problem of \cref{thm:bounded-inf}, even for $M=1$ and some fixed
$\eps>0$, would improve the state of the art for minimum-witness convolution.
\end{remark}

\section{Research History and Use of AI Tools}
\label{sec:history}

Since October 2025, we have been working toward a truly subcubic algorithm for APSP, and toward truly subcubic and
truly subquadratic algorithms for the many special cases of APSP and of the $(\min,+)$-product studied in
fine-grained complexity. Since $(\min,+)$-convolution reduces to the $(\min,+)$-product, and since its monotone and
bounded-difference variants are among the best-understood special cases, much of our effort has focused on monotone
$(\min,+)$-convolution and on understanding how much of its structure is really needed for fast algorithms.

The first artifact of this program is our SODA~2027 paper~\cite{HajiaghayiMS27}, which gives truly sublinear dynamic
algorithms for both $(\min,+)$-convolution and $(\min,+)$-matrix multiplication under marginal updates. Its central
contribution is a framework that carries the structural advantages of static, bounded-value problems over to a
dynamic setting in which the values may grow without bound. For $(\min,+)$-convolution, we combined this framework
with structure from additive combinatorics and obtained update and query times of $\Ot(n^{21/22})$; for
$(\min,+)$-matrix multiplication, we combined it with graph-theoretic structure specific to bounded matrices and
obtained truly sublinear algorithms whenever the updates are not too concentrated, which is tight under the OMv
conjecture. The paper also ties the dynamic and static problems together: static bounded-difference
$(\min,+)$-convolution reduces to its dynamic counterpart, so sufficiently fast dynamic algorithms would improve the
static exponent, and our algorithms yield solutions to the bounded-difference versions of the problems introduced by
Chan and Lewenstein~\cite{ChanL15} and by Bringmann, Grandoni, Saha, and Vassilevska
Williams~\cite{BringmannGSW16}.

Very recently, Alman and Vassilevska Williams~\cite{AlmanV26} obtained a truly subcubic algorithm for APSP with
polynomially bounded integer weights, answering the central question that motivated our program for this
setting. In light of their
result, we are posting the present paper as an artifact of our progress so far on the structured special cases,
which remain of independent interest: for inputs of small monotonicity, our algorithms achieve exponents far below
the exponent close to $2$ that follows from~\cite{AlmanV26}.

All results in this paper were obtained without AI, except two, which were obtained with AI models: the
reduction from bounded to bounded monotone $(\min,+)$-convolution (\cref{thm:mono-red} and
\cref{cor:four-thirds,cor:tradeoff}), and the reduction of the value range (\cref{thm:value-red} and
\cref{cor:value-lb}). AI tools were additionally used in writing the paper.

\bibliographystyle{plainurl}
\bibliography{refs}

\end{document}